%% file: EC26paper.tex
\documentclass[format=acmsmall, review=false]{acmart}

\usepackage{acm-ec-26}
\usepackage{booktabs} 
\usepackage[ruled]{algorithm2e} 

\SetAlFnt{\small}
\SetAlCapFnt{\small}
\SetAlCapNameFnt{\small}
\SetAlCapHSkip{0pt}
\IncMargin{-\parindent}

\usepackage{enumitem} %
\usepackage{graphicx}
\usepackage{color}

\newtoggle{version}
\toggletrue{version}

\newtoggle{comments}
\toggletrue{comments}

\newtheorem*{theorem*}{Theorem}

\newtheorem{theorem}{\sc Theorem}
\newtheorem{definition}{\sc Definition}  [section]
\newtheorem{lemma}{\sc Lemma}[section]

\newtheorem{proposition}[lemma]{\sc Proposition}

\newtheorem{assumption}{\sc Assumption}

\SetKwInput{KwData}{Input}
\SetKwInput{KwResult}{Output}
\input{header}

\setcitestyle{authoryear}

\title{Ordinal Lindahl Equilibrium for Voting}

\author{Haoyu Song, Thanh Nguyen}

\begin{abstract}

The core is a central concept in multi-winner social choice, ensuring that no coalition of voters can support an alternative outcome whose size or cost exceeds the group's share of the electorate. This idea originates from the Lindahl equilibrium in classical public goods theory. Yet Lindahl equilibria may fail to exist when voters have ordinal preferences over a finite set of outcomes and monetary transfers are not allowed. We introduce \emph{Lindahl Equilibrium with Ordinal Preferences (LEO)}, extending the equilibrium framework to discrete collective choice. Using LEO, we construct randomized outcomes that satisfy (approximate) core constraints for a probabilistic set of voters, while ensuring that each voter is represented with high probability. We also provide a deterministic approximate core guarantee with a factor of 6.24, improving on the previous bound of 32. In structured environments, these outcomes can be computed efficiently. Overall, our results extend classical equilibrium concepts, providing a normative foundation for  proportional representation and practical algorithms for applications in voting and fair machine learning.

\end{abstract}

\begin{document}

\begin{titlepage}

\maketitle

\vspace{1cm}
\setcounter{tocdepth}{2} 
\tableofcontents

\end{titlepage}

\section{Introduction}
\input{newintro}

\section{Model}\label{sec:model}
\input{newmodel}

\section{Lindahl Equilibrium with Ordinal Preferences}\label{sec:lindahl}
\input{newLEO}

\section{From LEO to a Randomized Core}\label{sec:random}

\input{LEOtorandom}

\section{From Randomized to Deterministic  Core}\label{sec:main}
\label{sec:deterministic}
\input{new_deterministic}



\section{Efficient Computation of Core Allocations}\label{sec:computation}
\input{newcomputation}

\section{Applications and Conclusions}\label{sec:conclude}

\input{conclusion}

\newpage
\bibliographystyle{ACM-Reference-Format}
\bibliography{ref}

\appendix
\input{newappendix1}

\input{proofs}

\end{document}

%% file: header.tex
\usepackage{amsmath}
\usepackage{graphicx}
\usepackage{tabularx}
\usepackage{comment}
\usepackage{xcolor}
\usepackage{url}
\usepackage{tikz,pgfplots}
\usepackage{amssymb}
\usetikzlibrary{shapes.geometric, arrows}
\usepackage{thm-restate}
\usepackage{natbib}
\usepackage{float}
\usepackage{subcaption}
\usepackage{graphicx}
\usepackage[dvipsnames]{xcolor}

\newcolumntype{C}{>{\centering\arraybackslash}p{1em}}

\newcommand{\Xomit}[1]{}

\newcommand{\randbud}[1]{\mathcal{I}_{#1}}
\newcommand{\z}{{\bf z}}
\newcommand{\y}{{\bf y}}
\newcommand{\x}{{\bf x}}
\newcommand{\prices}{{\bf p}}

\usepackage{algorithmic}

\hypersetup{
     colorlinks=true,
     linkcolor=blue,
     filecolor=blue,
     citecolor = blue,      
     urlcolor=cyan,
     }

\def\final{0}  
\def\iflong{\iffalse}
\ifnum\final=0  
\newcommand{\thanh}[1]{{\color{blue}[{\small Thanh: \bf #1}]\marginpar{\color{red}*}}}
\newcommand{\ysl}[1]{{\color{brown}[{\small Young-San: \bf #1}]\marginpar{\color{red}*}}}
\newcommand{\haoyu}[1]{{\color{olive}[{\small Haoyu: \bf #1}]\marginpar{\color{red}*}}}
\else 
\newcommand{\thanh}[1]{}
\newcommand{\ysl}[1]{}
\newcommand{\haoyu}[1]{}
\fi  

%% file: newintro.tex
Multiwinner selection is a voting problem in which a group of agents report their preferences over alternatives, and the goal is to select a collective outcome that represents these agents fairly. Two prominent examples are committee selection and participatory budgeting (PB). In PB, each project has a cost, and voters express ordinal preferences over combinations of projects; the selected subset must fit within the budget while respecting voter preferences. Committee selection is a special case where all projects have equal cost (typically normalized to~1) and the budget corresponds to the size of the committee. Informally, proportional representation requires that no group of voters can support an alternative outcome whose size (or cost) exceeds their share of the electorate. This principle is captured by the \emph{proportional core} \citep{EJR2017,fain2016core}, which is motivated by the \emph{Lindahl equilibrium} and the core of economies with divisible public goods.


The traditional Lindahl equilibrium provides a canonical benchmark for public good provision by replacing a uniform price with personalized prices tailored to individual participants. At a Lindahl equilibrium, each agent chooses their preferred level of the public good given their personalized price and income, and equilibrium requires that all agents demand the same outcome. This coincidence of demand is coordinated by a fictitious central agent that selects the outcome to maximize total revenue---the sum of payments collected from all agents---relative to the cost of provision. By aligning payments with individual valuations, the Lindahl framework offers an appealing model of efficient and equitable public good provision. In particular, \citet{Lindahl70} shows that Lindahl equilibria lie in the core.

Despite its conceptual appeal, neither Lindahl equilibria nor the proportional core generally exist in voting environments. The difficulty stems from the inherently non-convex nature of these settings: feasible outcomes are discrete, and voters typically have ordinal preferences over alternatives or subsets of alternatives. As a result, the Lindahl framework does not directly extend to many collective decision-making problems of practical interest.

Prior work has therefore primarily focused on computationally tractable rules and on establishing fairness guarantees within this class. Such approaches often relax proportional-core requirements, restrict the domain of voter preferences, or combine both strategies. This \emph{bottom-up} perspective prioritizes implementable rules and tailors fairness notions to the practical constraints of these mechanisms.

In contrast, we adopt a \emph{top-down} approach. Rather than being constrained by computational considerations, we take the equilibrium solution as the normative benchmark. We then extend it to ordinal preference settings to construct outcomes that approximately satisfy core constraints. This approach provides an economic framework for achieving fair and proportional outcomes, even in discrete combinatorial settings. 

Specifically, we consider a general model of social choice, which we call \emph{budgeted social choice}, in which each outcome has an associated cost. The goal is to select an outcome whose cost does not exceed a given budget parameter and that (approximately) satisfies the core constraint. We assume a subadditive merge operation that combines two outcomes into a new, improved outcome with subadditive cost.
This framework generalizes both committee selection and participatory budgeting, where outcomes are subsets of candidates or projects and the merge operation corresponds to set union. It also captures emerging settings such as generative social choice \cite{FishGolzProcacciaRusakShapiraWuthrich2024}, motivated by AI and large language models (LLMs), in which outcomes are textual statements, the cost reflects their length (or complexity), and merging corresponds to combining or refining statements.

We make four main contributions.

First, we introduce the \emph{Lindahl Equilibrium with Ordinal Preferences (LEO)}, a modification of the classical Lindahl framework designed for environments with ordinal preferences. The modification has two components. Voters' incomes are drawn from continuous random distributions, which smooths aggregate demand. In addition, the standard requirement that individual demands exactly coincide with the central allocation is relaxed: agreement is enforced only for components with strictly positive prices. These changes allow us to establish existence using a fixed-point argument while retaining the efficiency and fairness properties associated with Lindahl equilibria. Our first result shows that a LEO exists whenever the income distribution has a continuous cumulative distribution function.

Second, we introduce the notion of a \emph{randomized core outcome}, parameterized by two quantities. A \((\lambda, \gamma)\)-randomized core is a lottery over outcomes in which, for each realization, only a subset of voters is represented. The \emph{$\gamma$-core constraint} requires that no group within this subset can deviate to an alternative outcome whose cost is less than $\gamma$ times their share of the electorate. Across the lottery, each voter belongs to a represented group with probability at least $\lambda$. In this sense, $\gamma$ measures the strength of the core constraint for included voters, while $\lambda$ measures the overall coverage across the electorate.

Using LEO, we show that for every $\alpha \ge 0$, there exists a $(1-e^{-\alpha}, \alpha+1)$-randomized core solution, making explicit the tradeoff between coverage (the probability that a voter is represented) and the strength of the core constraint.

It is known that a deterministic 2‑$\epsilon$-approximate core cannot be guaranteed for all voters \citep{ApproStable}. As an implication, setting $\alpha = 1$ yields a single outcome that satisfies the 2-approximate core for at least 63\% of voters. While this is already notable, such a deterministic outcome may leave the remaining 37\% of voters unrepresented. For $\alpha = 3$, coverage improves further: each voter is represented with probability at least $1 - 1/e^3 \approx 95\%$, while the outcome satisfies a 4-approximate core constraint. But importantly, randomization strengthens the guarantee for all voters: by constructing a lottery over outcomes, \emph{every voter} is represented with probability at least 63\%, avoiding systematic exclusion. In this way, the randomized core provides proportional representation across the entire electorate that no single deterministic outcome can achieve.

Third, we show how to convert the randomized core into a \emph{deterministic core solution}, which includes all voters. The resulting outcome achieves a $6.24$-approximation, a substantial improvement over the previous $32$-approximation in \cite{ApproStable}.

Lastly, regarding computation, our approach can be adapted for algorithmic purposes in certain scenarios. The key observation is that the equilibrium conditions underlying the construction of a randomized core can be expressed as a set of inequalities associated with the LEO. Under a specific income distribution—namely, the uniform distribution—these inequalities are linear. Consequently, the existence of a LEO guarantees that the corresponding linear program has a solution, which can be computed efficiently when the set of outcomes \(\mathcal{C}\) for which core constraints are enforced is of polynomial size. This allows us to construct both randomized and deterministic outcomes satisfying these constraints in time polynomial in \(|\mathcal{C}|\). In particular, our algorithm achieves a 11.6 approximation, improving upon the 32-approximation obtained through \cite{ChengJiangMunagalaWang2019,ApproStable}. The results also have practical implications, such as ranking ballot and applications in fair machine learning.

The paper is organized as follows. After reviewing related work, Section~\ref{sec:model} introduces the model and  solution concepts. Section~\ref{sec:lindahl} presents the Lindahl equilibrium and extends it to ordinal preferences through the introduction of LEO. Section~\ref{sec:random} discusses the construction of the randomized core, while Section~\ref{sec:main} presents methods for constructing an approximate deterministic core solution using LEO. Section~\ref{sec:computation} addresses the computational implementation  and applications of both the randomized and deterministic core. Section~\ref{sec:conclude} presents additional applications and concludes. All omitted proofs are provided in the Appendix.

\subsection*{Related Work}
This paper contributes to the literature on multiwinner selection in social choice, which has motivated several notions of fair representation. The strongest of these is the \emph{core}, requiring that no large enough subgroup of voters can block an outcome in favor of an alternative they collectively prefer. While the core provides proportional guarantees for all subgroups, it may fail to exist in many voting environments due to the discrete and combinatorial nature of feasible outcomes \citep{fain2018fair}. In contrast, \emph{justified representation} (JR) and its variants—such as EJR \citep{EJR2017}, EJR+ \citep{EJR+2023}, PJR \citep{additive21,PJR2022}, and FJR \citep{additive21}—ensure representation only for certain ``cohesive'' subgroups, and a comprehensive overview of the related notions can be found in the excellent surveys of \cite{pb_survey2020,pb_survey2023}. JR-based notions are often computationally tractable and natural for approval-based or additive-preference settings, but they offer weaker guarantees than the core and do not capture complementarities among candidates or sets, which can be expressed in ordinal rankings over committees. 
The conceptual appeal of the core, along with its frequent nonexistence under discrete, ordinal preferences, motivates the study of randomized and approximate core solutions.

\paragraph{\textit{Randomized Core}} Randomized core solutions have been well studied in multi-winner settings. For instance, \cite{aziz2023bestwine, aziz2024fair} consider committee voting and participatory budgeting, focusing on \emph{Justified Representation} rather than the core.  More closely related to our notion of a randomized core is the concept of a \emph{stable lottery} introduced by \cite{ChengJiangMunagalaWang2019}, with \cite{ApproStable} later establishing the existence of 2-stable lottery. A key difference is that in \cite{ChengJiangMunagalaWang2019}, stability is defined \emph{in expectation}: for any group, the expected fraction of members who prefer an alternative outcome is bounded, but there is no guarantee for individual realizations. As a result, some voters may be underrepresented in all realizations. By contrast, our $(\lambda, \gamma)$-randomized core enforces the core constraint \emph{within each realization}, with varying subgroups across the lottery, ensuring that representation is spread across the entire electorate.

\paragraph{\textit{Deterministic Approximate Core.}}
For deterministic approximate core solutions, the best known result for general monotone preferences is the $32$-approximation of \cite{ApproStable}, which relies on stable lotteries as described above. By contrast, we use the $(\lambda, \gamma)$-randomized core, which provides greater control over the parameters and can yield better approximations for ex-post deterministic solutions.

In specialized settings, additional structure on preferences allows for stronger guarantees or polynomial-time algorithms. The structure most closely related to our work is ranking ballots, where voters provide a ranking over candidates rather than subsets of candidates. For committee selection (where each candidate has unit cost) under ranking ballots, the first approximation guarantee for the core was $16$ \citep{ChengJiangMunagalaWang2019, ApproStable}, later improved to $9.82$ \citep{six}, while a lower bound of $2$ is known.  When candidates have heterogeneous costs, the best known existence bound is $32$ \citep{ApproStable}, identical to the general case with monotone preferences. Regarding computation, \cite{ApproStable} also implies  a polynomial-time algorithm for computing a 32-approximate core solution, which our method improves to 11.6.


Another important structure is approval ballots in committee selection, where voters provide a list of approved candidates. The existence of an exact core solution in this setting remains an open and important question.  \cite{ChengJiangMunagalaWang2019} and \cite{ApproStable} together established the existence of a 16-approximate core solution, and \cite{gao2025computation} improved the approximation ratio to 3.65. There are also stronger results in restricted or related settings. For example, when the committee size is small (fewer than 8), \cite{few_seats25} shows that an exact core solution exists.


\paragraph{\textit{Alternative Approximate Core Concepts.}}
Other works have proposed relaxed notions of the approximate core, such as the $(\alpha, \beta)$-core \citep{fain2018fair, peters2020proportionality, munagala2022approximate, multi23}, which permits blocking coalitions to achieve $\alpha$-fold utility gains and to exceed standard core constraints by a factor of $\beta$. These methods rely on cardinal valuations and weaken individual rationality (IR). In contrast, our approach preserves ordinal preferences and strictly enforces IR.

\paragraph{\textit{Lindahl equilibrium approach}} 
The Lindahl equilibrium is a classical concept in public goods and market design with externalities, originating from early economic theory \citep{foley1966resource,Lindahl70,rader1973economic} and more recently developed in both economics \citep{gul2020lindahl} and computer science \citep{fain2016core,kroer2025computing,gao2025computation}. Existing approaches, however, rely on cardinal preferences, since establishing the existence of a Lindahl equilibrium typically requires continuity, which in turn depends on cardinal utility. In contrast, we are the first to formalize this framework using only ordinal preferences. Follow-up work such as \cite{SongNguyenLin2026} apply this concept to study the Condorcet-winning set and its generalizations.

%% file: newmodel.tex
\subsection{Budgeted Social Choice}
For full generality and ease of notation, we use a more abstract model. There is a finite set of social outcomes $\mathcal{O}$, with each outcome $o\in \mathcal{O}$ having a cost $c(o)\geq 0$. We assume that $\mathcal{O}$ contains the outcome $\perp$ with $c({\perp}) = 0$, corresponding to the outside option. We also assume $\perp$ is the only outcome with 0 cost. And we assume that the smallest cost of any outcome other than $\perp$ is 1.

Given a budget $B$, a \emph{feasible outcome} is an outcome in $\mathcal{O}$ with a cost of at most $B$.  A lottery over $\mathcal{O}$ is a discrete random variable that takes values in $\mathcal{O}$. It can be represented as a vector in $[0,1]^{|\mathcal{O}|}$ whose coordinates sum to 1. The set of all such lotteries is denoted by $\Delta(\mathcal{O})$. We use a tilde over a capital letter, such as $\widetilde{O}$, to denote an element of $\Delta(\mathcal{O})$. We also refer to a lottery over \( \mathcal{O} \) as a probabilistic or randomized outcome.

Let \(N\) denote the set of agents (voters), with \(|N| = n\). Each agent \(i\) has a strict preference order \(\succeq_i\) over \(\mathcal{O}\), with \(\perp\) as the least-preferred outcome.


The tuple \((\mathcal{O}, c(\cdot), \{\succeq_i\}_{i=1}^n),B)\) is called an instance of \emph{budgeted social choice}.

We make the following assumption about the outcome space and the preferences of voters.

\begin{assumption}[Sub-additive Merging Operation]\label{assu:merging}
There exists a merging operation $\oplus$ s.t. for any  two outcomes $o,o' \in \mathcal{O}$, 
it satisfies
$c(o\oplus o')\leq c(o)+c(o')$ and $o\oplus o'\succeq_i o,  o'$ for every  voter $i$.  
\end{assumption}

Participatory budgeting with monotone  preferences (and, as a special case, committee selection) fits naturally within our framework. In this setting, outcomes are subsets of projects (or candidates), and monotonicity implies that agents weakly prefer supersets to subsets. The cost of an outcome is the total cost of the projects included in the set (or, in the case of committee selection, the cardinality of the set), and merging outcomes corresponds to taking unions.


Another application arises in machine learning, particularly in clustering and classification. In this context, an outcome is a set of central points or labels used to organize data, and each data point has preferences over outcomes based on how well the outcome represents it. Merge operations can correspond to the union or combination of labels; for example, ``long fruit'' $\oplus$ ``yellow fruit'' could be combined into ``banana''. 

A related application has emerged in AI under the framework of \emph{generative social choice} \citep{FishGolzProcacciaRusakShapiraWuthrich2024}, where outcomes are textual statements generated by a system. Each statement incurs a cost reflecting generation resources (e.g., number of tokens), and merging corresponds to combining two statements into one. Subadditive costs are a natural assumption in this setting, since shared structure or content can reduce the total generation cost.

\subsection{Solution Concepts}
\begin{definition}[Core]
    Given an instance of  budgeted social choice  \((\mathcal{O}, c(\cdot), \{\succeq_i\}_{i=1}^n),B)\), a feasible  outcome $o^*$, i.e, $c(o^*)\le B$, lies in the core if, for every  outcome $o\in \mathcal{O}$, the number of agents who prefer $o$ to $o^*$ is at most $\frac{c(o)}{B}n$. That is, 
    $
    |\{i\in N : o^* \prec_i o\}| \leq \frac{c(o)}{B}n.
    $
\end{definition}

Given that the core can be empty, we will adopt the following  notion of approximate core.

\begin{definition}[$\gamma$-Core]\label{def:approximate}
    Given an instance of  budgeted social choice  \((\mathcal{O}, c(\cdot), \{\succeq_i\}_{i=1}^n),B)\),   a feasible outcome $o^*$ is in  \underline{$\gamma$-approximate core} if for every  outcome $o\in \mathcal{O}$ with $c(o)\le B$, 
    $
    |\{i\in N : o^* \prec_i o\}| \leq  \gamma\cdot \frac{c(o)}{  B}\cdot n.
    $
\end{definition}

Another useful way to define an approximate core is the following.

\begin{definition}[$\gamma$-Partial Core]\label{def:approximate1}
Given an instance of budgeted social choice \((\mathcal{O}, c(\cdot), \{\succeq_i\}_{i=1}^n, B)\), a pair \((o^*, V)\) consisting of a feasible outcome $o^*$ and a subset of voters $V \subseteq N$ is in the \underline{$\gamma$-partial core} if, for every outcome $o \in \mathcal{O}$ with $c(o)\le B$,
\[
|\{i \in V : o^* \prec_i o\}| \leq \gamma \cdot \frac{c(o)}{B} \cdot n.
\]
Given a voter $i$, we say that $i$ \underline{is represented} by the partial $\gamma$-core solution $(o^*, V)$ if $i \in V$.
\end{definition}

Compared with the $\gamma$-core, the partial $\gamma$-core guarantees the $\gamma$-core condition only for a subset of voters. Intuitively, a voter is represented if they belong to a group that cannot jointly deviate to an alternative outcome whose cost is proportional to the group’s size. Thus, voters in $V$ are protected against such proportional deviations, and therefore the outcome $o^*$ ``represents'' their preferences, while voters outside $V$ receive no such guarantee.

Consequently, voters outside $V$ may unanimously prefer an alternative outcome whose cost lies well within their proportional share of the electorate. This limitation motivates the following randomized solution, which mitigates this potential misrepresentation by spreading representation guarantees across voters.

\begin{definition}[$(\lambda,\gamma)$- Randomized Core]\label{def:approximate2}
Given an instance of budgeted social choice \((\mathcal{O}, c(\cdot), \{\succeq_i\}_{i=1}^n, B)\), a distribution over partial $\gamma$-core solutions is a \underline{$(\lambda,\gamma)$-randomized core} solution if, for every voter $i \in N$, the probability that $i$ is represented by the realized $\gamma$-partial  core solution is at least $\lambda$.
\end{definition}

The  $(\lambda,\gamma)$-randomized core provides representation guarantees on average across outcomes. Any single partial $\gamma$-core solution may represent only a subset of voters, but the lottery ensures that each voter is represented with probability at least $\lambda$. In other words, while a voter may not be included in every outcome, they have a fair chance of being represented over the randomization. 

The parameter $\gamma$ controls how strong the representation is for voters who are included—it limits how large and influential a group can be while still preferring an alternative outcome. The parameter $\lambda$ measures how widely this representation is spread across the population.
Unlike the deterministic $\gamma$-core, the randomized framework allows a natural tradeoff between these two dimensions of fairness. Stronger representation for included voters (smaller $\gamma$) may require focusing on fewer voters, leading to a smaller $\lambda$. Conversely, representing more voters (larger $\lambda$) may require relaxing the strength of representation (larger $\gamma$). This tradeoff captures the balance between the intensity and the breadth of fairness, and we study which combinations of $(\lambda,\gamma)$ can be achieved.

%% file: newLEO.tex
This section introduces our main technical tool.  We adapt the traditional Lindahl equilibrium  framework  to the setting with ordinal preferences, yielding a fractional solution.

\subsection{Classical Lindahl equilibrium}
To better understand our new notion of Lindahl equilibrium with Ordinal preferences, and how it helps us to construct randomized and approximate  core solutions,  we first revisit the classical version in a convex economy with a continuous set of alternatives. While there are several equivalent formulations, we adopt the following for simplicity: every $\x \in \mathbb{R}_+^m$ is assumed to be a feasible alternative. Each agent $i$ has a \emph{strictly increasing and concave} utility function $u_i(\x)$ and is endowed with a fixed amount of token income, normalized to $1$. Given a personalized price vector $\mathbf{p}_i \in \mathbb{R}_+^m$, the agent chooses a bundle from their demand set:
$
D_i(\mathbf{p}_i) = \arg\max_{\x \in \mathbb{R}_+^m, \, (\mathbf{p}_i)^\top \x \leq 1} u_i(\x).
$

A centralized producer, indexed by $0$, has a strictly increasing \emph{linear} cost function $c(\x)$  and chooses an allocation $\x \in \mathbb{R}_+^m$ to maximize total revenue subject to a budget constraint $B>0$. The producer’s choice comes from the set:
$
D_0(\mathbf{p}_1, \ldots, \mathbf{p}_n) := \arg\max_{\x \geq 0} \left( \sum_{i=1}^n \mathbf{p}_i \right)^\top \x \quad \text{s.t.} \quad c(\x) = B.
$

In a Lindahl equilibrium, all agents and the producer agree on a common outcome $\x^*$, with personalized prices ensuring individual optimality and feasibility.

\begin{definition}[Lindahl Equilibrium]
A \emph{Lindahl equilibrium} consists of an allocation $\x^* \in \mathbb{R}_+^m$ and personalized price vectors $(\mathbf{p}_1, \ldots, \mathbf{p}_n) \in (\mathbb{R}_+^m)^n$ such that:
\begin{enumerate}
    \item \textbf{Individual optimality:} For each agent $i \in \{1, \ldots, n\}$, \;\;
    $
    \x^* \in D_i(\mathbf{p}_i).
    $
    \item \textbf{Producer optimality:} The producer chooses 
    $
    \x^* \in D_0(\mathbf{p}^1, \ldots, \mathbf{p}^n).
    $
\end{enumerate}
\end{definition}

\begin{theorem*}[\cite{Lindahl70}] \label{theo:Lindahl}
A Lindahl equilibrium exists, and the equilibrium allocation $\x^*$ lies in the core. That is, for any alternative $\x$,
\[
|\{ i : u_i(\x) > u_i(\x^*) \}| < c(\x) \cdot \frac{n}{B}.
\]
\end{theorem*}

\begin{proof}[Sketch]
Existence follows from a standard fixed-point argument: the demand correspondence is continuous in prices, utilities are concave, and costs are linear.

Let $\x^*$ be a Lindahl equilibrium. Since utilities are strictly increasing, each agent exhausts their income, so $\mathbf{p}_i \cdot \x^* = 1$ for all $i$, and thus $\sum_{i=1}^n \mathbf{p}_i \cdot \x^* = n$.

Linearity of costs implies that the producer maximizes the profit-to-cost ratio, yielding, for any alternative $\x$,
\[
\sum_{i=1}^n \mathbf{p}_i \cdot \x \le c(\x)\cdot \frac{n}{B}.
\]

Let $I=\{i : u_i(\x)>u_i(\x^*)\}$. By individual optimality, any $i\in I$ must satisfy $\mathbf{p}_i\cdot \x>1$. Hence,
\[
|I|<\sum_{i\in I}\mathbf{p}_i\cdot \x \le \sum_{i=1}^n \mathbf{p}_i\cdot \x \le c(\x)\cdot \frac{n}{B},
\]
which proves the claim.
\end{proof}

The argument relies on three structural properties of the Lindahl framework. First, continuity of demand ensures equilibrium existence via a fixed-point argument. Second, strictly increasing utilities imply that agents exhaust their income at equilibrium, aligning aggregate spending with the producer’s revenue. Third, linear costs imply that the centralized planner maximizes the profit-to-cost ratio, which yields the core guarantee. 
In our setting, with a discrete set of outcomes and ordinal preferences, these properties generally fail, creating the main challenges in extending the Lindahl equilibrium. The next section introduces an ordinal analogue that overcomes these difficulties.

\subsection{Lindahl Equilibrium with Ordinal Preference (LEO)}

Adapting the classical Lindahl equilibrium concept to our setting—with ordinal preferences and a discrete set of outcomes—requires expanding the outcome space from discrete to continuous. A natural starting point is to consider lotteries over outcomes, represented by $\Delta(\mathcal{O})$. However, as we will show below, this is insufficient for adapting the proof above. Instead, we work with outcomes represented by vectors $\y \in \mathbb{R}_+^{|\mathcal{O}\setminus\{\perp\}|}$, which can be interpreted as fractional allocations over the discrete alternatives.

Even with this relaxation of the outcome space, significant challenges remain in formulating a Lindahl-type equilibrium relying only on ordinal preferences. We address each challenge as follows:
\textbf{Lack of continuity:} In classical settings, continuous demand enables fixed-point arguments. In our  discrete setting with ordinal preferences, this continuity fails. To address this, we adopt the \emph{random income} method by assuming incomes are drawn from a continuous distribution \(\mathcal{I}\) on \([0,1]\).\footnote{\cite{randomincome24} also use this method in assignment problems. Our approach applies to any bounded, continuous distribution; we use \([0,1]\) for convenience.} For each income realization, agents choose their most preferred affordable outcome, with \(\perp\) (priced at zero) ensuring nonempty demand. This induces a random demand that varies continuously with prices, restoring the needed continuity.

\textbf{Satiated utilities:} In classical models, strictly increasing utilities ensure that all agents fully exhaust their income and consume the same allocation, which also coincides with the producer's output. This fails in our setting because utilities saturate over a finite outcome set. Instead, we adopt a condition akin to competitive equilibrium: each agent’s randomized demand \( \x_i \in \mathbb{R}_+^{|\mathcal{O}|} \) must not exceed the producer’s allocation \( \mathbf{y} \) coordinate-wise. If \( x_{i,o} < y_o \) for some outcome \( o \), then \( p_{i,o} = 0 \). This weaker condition allows for heterogeneous demands across agents, which is necessary to support a fixed-point argument, and still suffices to ensure core-like guarantees.

\textbf{Revenue maximization and the core:} Classical core arguments rely on the producer maximizing the revenue-to-cost ratio. But restricting outcomes to lotteries over \( \mathcal{O} \) (i.e., \( \sum_{o} y_o = 1 \)) blocks this approach. To recover it, we relax the outcome space to \( \mathbf{y} \in \mathbb{R}_+^{|\mathcal{O}|} \) and let the producer maximize revenue subject to \( \sum_{o\in \mathcal{O}\setminus \{\perp\}} y_o c(o) = B \). This retains the revenue-to-cost ratio logic but yields a fractional outcome, which must be rounded to a valid lottery. We use dependent rounding and this is where we incur a constant-factor approximation loss.

These modifications naturally lead to a new solution concept, which we call \emph{Lindahl equilibrium with Ordinal Preferences (LEO)} and formally define as follows.

To begin with, each voter $i$ is endowed with a \emph{random income} $\randbud{}$, supported on $[0,1]$ (the same for all voters, though this assumption can be easily relaxed) and has a  \textit{personalized} price vector $\prices_i \in R_{+}^{|\mathcal{O}|}$. We use $p_{i,o}$ to denote the personalized price of $o\in \mathcal{O}$ for an agent $i$  with $p_{i,\perp} = 0$.

Given a fixed deterministic income, the demand of a voter is her most preferred affordable outcome, that is $D_i(\prices_i,b):=\max_{\succ_i} \{ o \in \mathcal{O} : p_{i,o} \leq b \}$. Under a random income distribution \( \randbud{} \), voter \( i \)'s \textit{random demand} is the distribution over such choices across income realizations:
\[
\mathcal{D}_i(\prices_i, \randbud{}) := \left\{ \max_{\succ_i} \{ o \in \mathcal{O} : p_{i,o} \leq b \} \mid b \sim \randbud{} \right\}.
\]

Under the assumption that $\mathcal{I}$ is supported on $[0, 1]$ and $p_{i,\perp} = 0$, $\perp$ is affordable for every possible realization of the random income and therefore $\mathcal{D}_i$ is guaranteed to be a valid lottery over $\mathcal{O}$.

\paragraph{\bf {Example.}}
Suppose $a \succ_i b \succ_i c \succ_i \perp$, with prices $p_{i,a}=1/2$, $p_{i,b}=p_{i,c}=1/3$, and $p_{i,\perp} = 0$ and let the random income be $\randbud{}\sim \mathbf{U}[0,1]$. 
The agent consumes $a$ whenever the realized income is at least $1/2$, consumes $b$ whenever the realized income lies in $[1/3,1/2)$, and never consumes $c$. 
Thus, the resulting random demand is
\[
\mathcal{D}_i(\prices_i, \randbud{})= \frac{1}{2}a \;+\; \frac{1}{6}b \;+\; \frac{1}{3}\perp.
\]

The following continuity property of the random demand is crucial for establishing the existence of a LEO (see Appendix~\ref{app:conti}).

\begin{lemma}\label{lemma:conti} 
For every agent $i$ and outcome $o$, let $x_{io} = \Pr(\mathcal{D}_i(\prices_i, \randbud{}) = o)$. The vector $\x_i = \{x_{io}\}_{o \in \mathcal{O}}$ is continuous in $\prices_i$, provided that the cumulative distribution of $\randbud{}$ is continuous.
\end{lemma}


Similar to the standard Lindahl case, the producer aims to produce a fractional central allocation of outcomes maximizing the total profits given individual price vectors.  In particular, given a total budget constraint $B$, the demand of the producer is 
$$
\mathcal{D}_0(\prices_1,\prices_2,...,\prices_n) = \arg\max_{\z \in \mathbb{R}_+^{|\mathcal{O}\setminus \{\perp\}|}}  \sum_{o\in \mathcal{O}\setminus \{\perp\}}(\sum_{i=1}^n p_{i,o})\cdot  z_o \quad \text{s.t.} \quad \sum_{o\in \mathcal{O}\setminus \{\perp\}}c(o)z_o = B.
$$

Note that in this formulation, the producer only produces outcomes other than \( \perp \). This is purely without loss of generality: allowing production of \( \perp \) does not affect the producer's decision, since \( p_{i,\perp} = 0 \) and \( c(\perp) = 0 \).

Now we introduce the following definition of Lindahl equilibrium with ordinal preference (LEO).

\begin{definition} \label{def:LEO}
Given an instance of  budgeted social choice  \((\mathcal{O}, c(\cdot), \{\succeq_i\}_{i=1}^n),B)\), and a random income $\randbud{}$ supported on $[0,1]$, the Lindahl equilibrium with ordinal preference (LEO) consists of individual consumptions $\{\x_{i}\}_{i=1}^n \in [0,1]^{|\mathcal{O}|}$,  a common allocation $\y \in [0,B]^{|\mathcal{O}\setminus \{{\perp}\}|}$ and personalized prices $\prices_{i=1}^n \in [0,1]^{|\mathcal{O}|}$ with  $p_{i,\perp} = 0 $ such that  

\begin{enumerate} 

\item \label{leo-1} for any $i\in N$ and $o\in \mathcal{O}$, $x_{i,o}=  \Pr(\mathcal{D}_{i}(\prices_i, \randbud{}) =o)$, 

\item \label{leo-2} for any $i\in N$ and $o\in \mathcal{O} \setminus\{\perp\}$, 
$y_o \geq x_{i,o}$  with strict inequality only when $p_{i,o} = 0$, 
 
\item \label{leo-3} $y\in \mathcal{D}_0(\prices_1,\prices_2,...,\prices_n)$.

\end{enumerate}
\end{definition}



\begin{figure}[h!]
    \centering
    \includegraphics[width=0.75\linewidth]{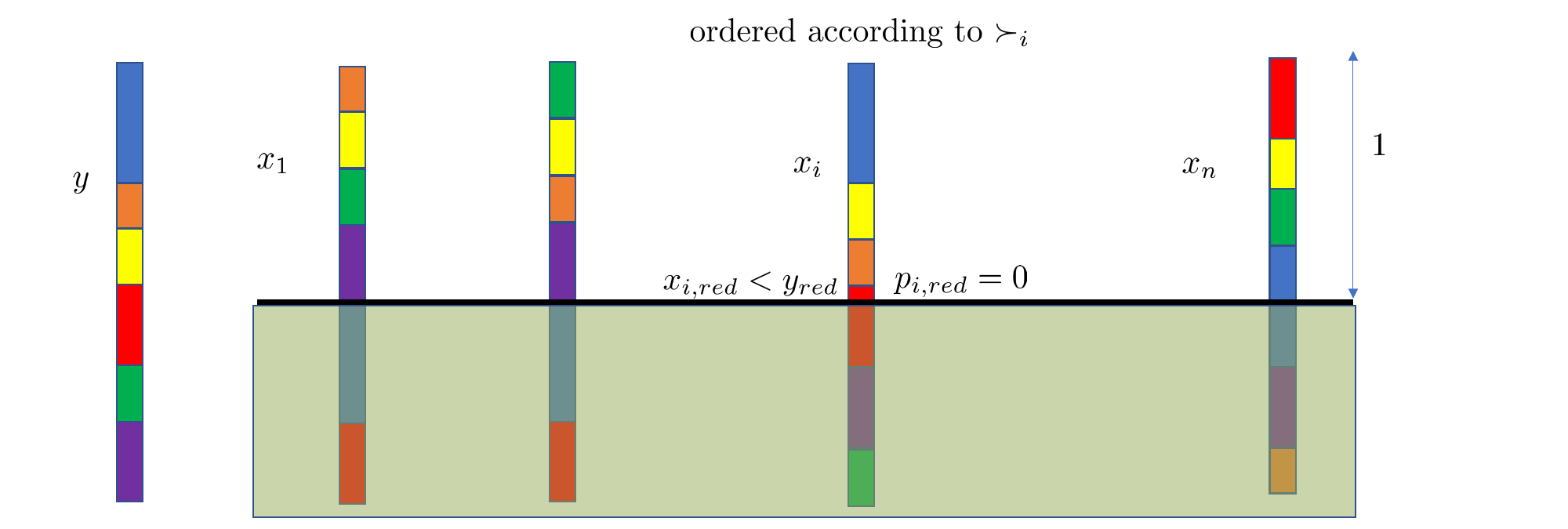}
    \caption{An Illustration of LEO}
    \label{fig:LEO}
\end{figure}
Figure~\ref{fig:LEO} illustrates a LEO. Each outcome is represented by a color, and the length of each colored interval corresponds to the allocated value $y$. For each agent $i$, the outcomes are arranged vertically in decreasing order of preference $\succ_i$. By construction of the random demand, more preferred outcomes receive higher prices, so the individual prices of these outcomes  are sorted in decreasing order along the vertical axis. Since each agent consumes a lottery, we have $\sum_{o\in \mathcal{O}} x_{io}=1$. Thus, $x_i$ specifies the portion of each multicolored bar above the horizontal line. Condition~\ref{leo-2} of Definition~\ref{def:LEO} implies that whenever the price is positive, the allocations $x$ and $y$ coincide. Consequently, any outcome that lies at or below the intersection point of $x$ and the horizontal line  has zero price.

{As discussed at the beginning of this section, the existence of LEO relies on two key observations. First, random demand is continuous whenever the income distribution has a continuous cumulative function. Second, because satiation complicates the usual fixed-point argument, we impose only a weaker requirement on individual consumption (condition~\ref{leo-2}) instead of full equality with the common allocation.} These observations allow us to establish the existence of a LEO from the standard Kakutani fixed-point theorem; the proof is given in Appendix~\ref{app:LEO}.


\begin{theorem}\label{thm:existence} 
Given an instance of  budgeted social choice  \((\mathcal{O}, c(\cdot), \{\succeq_i\}_{i=1}^n),B)\), if the cumulative distribution function \(F_{\randbud{}} : [0,1] \rightarrow [0,1]\) of \(\randbud{}\) is continuous, then a LEO exists.
\end{theorem}

Next, we introduce an additional parameter \( \alpha > 0 \) into the LEO framework. The only modification is in the second condition of LEO. The existence of an \(\alpha\)-LEO follows directly from Theorem~\ref{thm:existence}. This parameter will later be selected to optimize the approximation ratio in our final construction.

\begin{definition} \label{def:alpha-gamma-LEO}
Given an instance of  budgeted social choice  \((\mathcal{O}, c(\cdot), \{\succeq_i\}_{i=1}^n),B)\), a random income $\randbud{}$ supported on $[0,1]$ and $\alpha>0$, an $\alpha$-LEO consists of individual consumptions $\x_{i=1}^n \in [0,1]^{|\mathcal{O}|}$, a common allocation $\y \in [0,B]_+^{|\mathcal{O}\setminus\{\perp\}|}$ \footnote{Recall that we assume $c(o)\geq 1$ for any $o\neq \perp$.}, and prices $\prices_{i=1}^n \in [0,1]^{|\mathcal{O}|}$ with  $p_{i,\perp} = 0$, such that

\begin{enumerate} 
\item \label{alphaleo-1} for any $i\in N$ and $o\in \mathcal{O}$, $x_{i,o}=  \Pr(\mathcal{D}_{i}(\prices_i, \randbud{}) =o)$,

\item \label{alphaleo-2} for any $i\in N$ and $o\in \mathcal{O}\setminus \{\perp\}$, 
$y_o \geq  \alpha \cdot  x_{i,o}$ with strict inequality only when $p_{i,o} = 0$;

\item \label{alphaleo-3} $\y\in \mathcal{D}_0(\prices_1,\prices_2,...,\prices_n)$.
\end{enumerate}
\end{definition}

\begin{theorem}\label{theo:alphaLEO}
Given an instance of  budgeted social choice  \((\mathcal{O}, c(\cdot), \{\succeq_i\}_{i=1}^n),B)\) and  $\alpha >0$, if the cumulative distribution function \(F_{\randbud{}} : [0,1] \rightarrow [0,1]\) of \(\randbud{}\) is continuous, then an $\alpha$-LEO exists.
\end{theorem}

\begin{proof}  

The theorem easily follows from Theorem \ref{thm:existence}. Let a budget $B$, a random income $\mathcal{I}$ supported on $[0,1]$, and $\alpha>0$ be given. Then let $(\x_{i=1}^n,\y,\prices_{i=1}^n)$ be the LEO under budget $\frac{B}{\alpha}$ and random income $\mathcal{I}$. We claim $(\x_{i=1}^n,\alpha\cdot \y, \prices_{i=1}^n)$ is the desired $\alpha$-LEO. It is immediately clear that $(\x_{i=1}^n,\alpha\cdot \y,\prices_{i=1}^n)$ satisfies parts \ref{alphaleo-1} and \ref{alphaleo-2} of Definition \ref{def:alpha-gamma-LEO} because $(\x_{i=1}^n,\y, \prices_{i=1}^n)$ satisfies parts \ref{leo-1} and \ref{leo-2} of Definition \ref{def:LEO}. Furthermore, since $\y$ maximizes revenue given prices $\prices_{i=1}^n$ under budget $\frac{B}{\alpha}$, it follows that $\alpha\cdot \y$ maximizes revenue given prices $\prices_{i=1}^n$ under budget $B$. Therefore, part \ref{leo-3} of Definition \ref{def:alpha-gamma-LEO} is satisfied and this completes the proof. 
\end{proof}

%% file: LEOtorandom.tex

Our main theorem in this section is the following.
\Xomit{
\begin{theorem}\label{theo:random}
Given an instance of budgeted social choice \((\mathcal{O}, c(\cdot), \{\succeq_i\}_{i=1}^n, B)\), for any $\gamma \ge 0$, there exists a $(1-e^{-\gamma}, \gamma)$-randomized core solution such that the expected cost of the outcome is at most $B$, and each realized outcome has cost at most
\[
B+\max_{o: c(o)\le B} c(o) \le 2B .
\]
\end{theorem}
\thanh{maybe just keep this theorem as the main theorem in this section, get rid of Theorem~\ref{theo:random} }
If no budget violation is allowed at realizations, we scale down the budget \(B\) by a factor of \(\frac{\gamma}{\gamma+1}\) 
and, in the construction of LEO, consider only outcomes with cost at most \(B/(\gamma+1)\). This gives us:

}

\begin{theorem}\label{theo:random}
Consider an instance of budgeted social choice 
\((\mathcal{O}, c(\cdot), \{\succeq_i\}_{i=1}^n, B)\) satisfying Assumption~\ref{assu:merging}. 
Then, for any \(\alpha \ge 0\), there exists a 
\((1 - e^{-\alpha}, \alpha + 1)\)-randomized core.
\end{theorem}

Thus, setting $\alpha = 1$ corresponds to satisfying the 2-core constraint, which guarantees each voter a representation level of at least $1-\frac{1}{e} \approx 63\%$.  Increasing $\alpha$ relaxes the blocking power of coalitions while strengthening individual representation guarantees. For example, setting $\alpha = 3$ yields a $(1-e^{-3}, 4)$-randomized core solution, ensuring that each voter attains representation with probability at least $1-\frac{1}{e^{3}} \approx 95\%$.

This randomized core solution is constructed starting from a LEO allocation. Notice that the common allocation $y$ in an $\alpha$-LEO is a quantity vector over all outcomes, satisfying the feasibility constraint 
\(
\sum_{o\in \mathcal{O}} c(o) y_o = B.
\) 
It is possible that 
\(
\sum_{o\in \mathcal{O}} y_o > 1,
\) 
so $y$ does not necessarily correspond to a randomized outcome. To construct a randomized core,  by  \emph{merge operation} that combines coordinates of $y$ to obtain a lottery. We do so by well known dependent rounding. 

Specifically, from the common allocation $\y$ satisfying the feasibility constraint 
\(
\sum_{o\in \mathcal{O}} c(o) y_o = B.
\) 
we can construct a lottery over outcomes using dependent rounding. We use the following lemma.

\begin{proposition}\label{proposition:dep} \cite{dependent_general}
There exists a polynomial-time algorithm that, given a vector $\z \in [0,1]^m$ with $\sum_{k=1}^m a_k z_k \le B$, computes a distribution $\widetilde{Z}$ over $\{0,1\}^m$ satisfying the following properties.

\begin{itemize}
    \item preservation of marginals: $\underset{\bold{Z}\sim \widetilde{Z}}{E}[Z_k]= z_k$ for all $k\in [m]$;
    \item  preservation of weights up to one: with $\bold{Z}\sim \widetilde{Z}$, it always satisfies $\sum_{k=1}^m a_k\cdot Z_k \leq B + \max\{a_k:z_k>0\}$; 
\item negative dependence between entries: for any $S \subseteq [m]$ , 
$\underset{\bold{Z}\sim\widetilde{Z}}{\Pr}\big( \bigwedge_{k\in S}Z_k = 0 \big) \leq \underset{k\in S}{\prod}\big(1-z_k\big)$.
\end{itemize}
\end{proposition}

With Proposition \ref{proposition:dep} as a building block, we have the following algorithm turning a fractional solution into a valid lottery over outcomes. 

\smallskip
\begin{algorithm}[H]\label{alg:dependent-rounding}  
\caption{Dependent Rounding on Social Outcomes}
\KwData{A fractional $\y\in \mathbb{R}_{+}^{|\mathcal{O}|} $ s.t. $\sum_{o\in O}c(o)y_o\leq B$}
\KwResult{A lottery $\widetilde{L}(\y)$ over outcomes with the cost at most  $B + \max\{c(o):y_o>0\}$ }
For each $o \in \mathcal{O}$  set  $z_{o} := \min\{1, y_{o}\}$  \label{line1}\\
Generate a distribution $\widetilde{Z}$ from $z$ satisfying Proposition \ref{proposition:dep}\\
Sample $\mathbf{Z} \in \{0,1\}^{|\mathcal{O}|} \sim \widetilde{Z}$ \\
Output $O = \underset{o:Z_{o} = 1}{\bigoplus} o'$, where $\bigoplus$ denotes the merging operation for $\mathcal{O}$.
\end{algorithm}

\smallskip
To associate an outcome with a set of voters in order to form a partial core solution, we introduce the following  notion of $\tau$-\textit{covering}.


\begin{definition}
Let $\{\prices_i\}_{i=1}^n$ be a set of personalized vectors for voters. 
\begin{itemize}
    \item For voter $i$ and $\tau\in[0,1]$, the \emph{$\tau$-boundary outcome} $\overline{o}_{i,\tau}$ is her demand at income level of $\tau$ 
    \[
        \overline{o}_{i,\tau}
        \;:=\;
        \max_{\succ_i}\{\,o\in\mathcal{O} : p_{i,o}\le \tau\,\}.
    \]
    \item An outcome $o\in\mathcal{O}$ \emph{$\tau$-covers} voter $i$ if $o \succeq_i \overline{o}_{i,\tau}$.  Denote the set of voters that are $\tau$-covered by $o$ by $V^{\tau}(o,\prices)$.
    That is, $V^{\tau}(o,\prices):=\{i\in N: o \succeq_i \overline{o}_{i,\tau}\}$.
\end{itemize}
\end{definition}

Thus, for every outcome $o$ and cutoff $\tau$, the set of voters associated with $o$, denoted $V^{\tau}(o,\prices)$, consists of those who value $o$ sufficiently. In particular, when income is fixed at $\tau$, these voters either already demand $o$, or would demand it if its price were at most $\tau$.


\begin{definition}\label{def:dependent-lottery}  
Let $\y$ be an allocation with $\sum_{o \in C} c(o) y_o \le B$, and let $\widetilde{L}(\y)$ be the lottery over outcomes produced by Algorithm~\ref{alg:dependent-rounding}. We call $\widetilde{L}(\y)$ the \emph{Dependent-Lottery} induced by $\y$.

Given personalized prices $\prices_{i=1}^n$, associate each realization $o \sim \widetilde{L}(\y)$ with its $\tau$-covered set $V^\tau(o)$. The resulting joint distribution over outcome–voters pairs is denoted $\widetilde{L}^\tau(\y, \prices)$.
\end{definition}

The following result is the main technical result of this section, which we use to construct the randomized core.

\begin{lemma}\label{theo:maintech}
Given an instance of budgeted social choice satisfying Assumption~\ref{assu:merging}, let $(\x_{i=1}^n, \y, \prices_{i=1}^n)$ be an $\alpha$-LEO with random income $\randbud{}$. For any cutoff parameter $\tau \in (0,1]$, the lottery $\widetilde{L}^{\tau}(\y,\prices)$ defines a distribution over pairs $(o^*, V)$ consisting of an outcome and a set of voters, such that:

\begin{enumerate}[label=(\roman*)]
    \item For every voter $i$, the probability that $i \in V$ is at least
    \[
    \lambda = 1 - e^{-\alpha \cdot \Pr(\randbud{} \ge \tau)}.
    \]\label{cnd1}
    
    \item For every realization $(o^*,V)$ and any outcome $o \in \mathcal{O}$,
    \[
    \left| \{ i \in V : o^* \prec_i o \} \right|
    \le
    \gamma \cdot \frac{c(o)}{B} \cdot n,
    \qquad
    \text{where }
    \gamma = \frac{\alpha \cdot \mathbb{E}[\randbud{}]}{\tau}.
    \] \label{cnd2}
    
    \item For every realization $(o^*,V)$,
    \[
    c(o^*) \le B + \max_{o \in \mathcal{O}} c(o).
    \] \label{cnd3}
\end{enumerate}

\end{lemma}

 Lemma~\ref{theo:maintech} shows that $\widetilde{L}^{\tau}(\y,\prices)$ is ``almost'' a $(\lambda,\gamma)$-randomized core. The only remaining issue is feasibility, as a realized outcome from this lottery may exceed the budget $B$. We next show that this can be addressed by a simple scaling argument, from which Lemma~\ref{theo:maintech} readily implies Theorem~\ref{theo:random}.

\begin{proof}[Proof of Theorem \ref{theo:random}]
Fix a sufficiently small $\varepsilon \in (0,1)$ such that, for every $o \in \mathcal{O}$,
\[
\left\lfloor (\alpha +1 )\cdot \frac{c(o)}{B} \cdot n \right\rfloor
=
\left\lfloor \frac{\alpha+1}{1-\varepsilon} \cdot \frac{c(o)}{B} \cdot n \right\rfloor,   \qquad \text{ and define } \tau := 1-\varepsilon.
\]

Consider the subset of outcomes and a corresponding budget:
\[
\mathcal{O}^* := \{ o \in \mathcal{O} : c(o) \le \tfrac{1}{\alpha+1} B \}, 
\quad 
B^* := \frac{\alpha}{\alpha+1} B.
\]

Let $(\x,\y,\prices)$ be an $\alpha$-LEO on $\mathcal{O}^*$ with budget $B^*$ and random income $\randbud{} = \mathbf{U}[1-\varepsilon,1]$. 
Thus, we scale down the budget and consider an $\alpha$-LEO on a subset of outcomes.
We show that $\widetilde{L}^{\tau}(\y,\prices)$ is a $(1-e^{-\alpha}, \alpha+1)$-randomized core.

First, for any outcome realization $o^*$, by condition \ref{cnd3} of Lemma~\ref{theo:maintech},
\[
c(o^*) \le B^* + \max_{o \in \mathcal{O}^*} c(o) \le \frac{\alpha}{\alpha+1} B + \frac{1}{\alpha+1} B = B.
\]

Second, since $\randbud{} = \mathbf{U}[1-\varepsilon,1]$, we have $\Pr(\randbud{} \ge 1-\varepsilon) = 1$. By condition \ref{cnd1}, each voter $i$ belongs to $V$ with probability at least
\[
\lambda = 1 - e^{-\alpha \cdot \Pr(\randbud{} \ge \tau)} = 1 - e^{-\alpha}.
\]

Third, by condition \ref{cnd2} and noting $\mathbb{E}[\randbud{}]\le 1$, for any realization $(o^*,V)$  of the lottery $\widetilde{L}^{\tau}(\y,\prices)$ and  every $o \in \mathcal{O}^*$,
\[
\left| \{ i \in V : o^* \prec_i o \} \right| 
\le \frac{\alpha \cdot \mathbb{E}[\randbud{}]}{1-\varepsilon} \cdot \frac{c(o)}{B^*} \cdot n
\le \frac{\alpha+1}{1-\varepsilon} \cdot \frac{c(o)}{B} \cdot n.
\]
On the other hand, for $o \notin \mathcal{O}^*$, $c(o) > \frac{1}{\alpha+1} B$, so $(\alpha+1)\frac{c(o)}{B} \, n \ge n$, and the core constraint is trivially satisfied.

Thus, these three conditions show that $\widetilde{L}^{\tau}(\y,\prices)$ is a 
\(
\Bigl(1-e^{-\alpha}, \frac{\alpha+1}{1-\varepsilon}\Bigr)\text{-randomized core solution.}
\)
To eliminate the extra $\varepsilon$ factor, observe that for any realized $(o^*,V)$ and alternative outcome $o$, our choice of $\varepsilon$ implies
\[
\bigl|\{i \in V : o^* \prec_i o\}\bigr|
\leq 
\left\lfloor \frac{\alpha+1}{1-\varepsilon} \cdot \frac{c(o')}{B} \cdot  n \right\rfloor
=
\left\lfloor (\alpha+1) \cdot \frac{c(o)}{B} \cdot n \right\rfloor
\leq 
(\alpha+1)  \cdot \frac{c(o)}{B} \cdot n.
\]
Therefore, $\widetilde{L}^{\tau}(\y,\prices)$ is an $(1-e^{-\gamma},\gamma)$ randomized core solution, as claimed.

\Xomit{
\smallskip

To prove Theorem \ref{theo:random_no_violation}, let
\[
\mathcal{O'} := \{ o \in \mathcal{O} : c(o) \le \tfrac{1}{\gamma+1} B \},
\quad 
B' := \frac{\gamma}{\gamma+1} B,
\]
and consider a \(\gamma\)-LEO defined on \(\mathcal{O'}\) with the budget $B'$. Then the cost of any realization satisfies
\[
c(\text{realization}) \le B' + \max_{o \in \mathcal{O'}} c(o) \le B.
\]

Moreover, for any \(o' \in \mathcal{O'}\), the core constraint with respect to \(B'\) holds:\footnote{Again we can choose $\epsilon$ small enough for the equality to hold.}
\[
\bigl|\{i \in V^{\tau}(o, \prices) : o \prec_i o'\}\bigr|
\le \left\lfloor \gamma n \cdot \frac{c(o')}{B'} \cdot \frac{1}{1-\varepsilon} \right\rfloor
= \left\lfloor (\gamma +1) \frac{c(o')}{B} \, n \right\rfloor
\le (\gamma+1) \frac{c(o')}{B} \, n.
\]
}

\end{proof}

It remains to prove  Lemma~\ref{theo:maintech}.
\subsection{Proof of Lemma~\ref{theo:maintech}}

We proceed via two structural lemmas. The first lemma establishes that for any $\alpha$-LEO, the common allocation $\y$ and prices $\prices$ satisfy a set of key parametrized inequalities (Lemma~\ref{lemma:feasible}). The second lemma shows that whenever $(\y,\prices)$ satisfies these inequalities, the associated $\widetilde{L}^\tau(\y,\prices)$ induces a randomized solution with appropriate parameters (Lemma~\ref{lemma:main}).

%

First, suppose $(\x_{i=1}^n,\y,\prices_{i=1}^n)$ is an $\alpha$-LEO with random income distribution $\randbud{}$. We derive the set of (in)equalities that $\y$ and $\prices$ must satisfy.

The first inequality represents the capacity constraint.
\begin{equation}\label{eq:000}
    \sum_{o\in \mathcal{O}} c(o)y_o=B.
\end{equation}

Second, for any outcome $o$, the total price is bounded by

\begin{equation}\label{eq:001}
\sum_{i\in N} p_{i,o}
\le 
\alpha\cdot \frac{c(o)}{B}\cdot n \cdot \mathbb{E}(\mathcal{I})
\qquad \forall o \in \mathcal{O} .
\end{equation}
This holds because, in an $\alpha$-LEO, the centralized agent selects a revenue-maximizing solution subject to  \eqref{eq:000}.
Consequently, at equilibrium, for every outcome $o$, the ratio between its revenue and cost is bounded above by the total revenue attainable by the centralized agent divided by $B$. Moreover, the total revenue attainable by the centralized agent is at most $n \cdot \alpha \cdot \mathbb{E}(\randbud{})$. A formal proof is provided in Proposition~\ref{prop:upper_bound_price} in the Appendix.

Third, for any cutoff $\tau \in [0,1]$ and any voter $i$, recall that the \emph{boundary outcome} $\overline{o}_{i,\tau}$ is the most-preferred outcome that voter $i$ can afford at a price of at most $\tau$. We provide a lower bound on the total probability mass assigned by $\y$ to outcomes that voter $i$ weakly prefers to this boundary outcome. In particular,

\begin{equation}\label{eq:002}
\forall \tau \in [0,1], i\in N: \qquad
\sum_{o' \in \mathcal{O} : o' \succeq_i \overline{o}_{i,\tau}} y_{o'}
\ge
\alpha \Pr(\randbud{} \ge \tau)
\qquad \text{for } \overline{o}_{i,\tau} \neq \perp. 
\end{equation}



Observe that when $\overline{o}_{i,\tau} \neq \perp$, any realization of income of at least $\tau$ allows voter $i$ to afford $\overline{o}_{i,\tau}$ and therefore to demand an outcome weakly preferred to it. Hence,
\[
\sum_{o' \in \mathcal{O}:\, o' \succeq_i \overline{o}_{i,\tau}} x_{i,o'}
\ge
\Pr(\randbud{} \ge \tau).
\]
Since $(\x,\y,\prices)$ is an $\alpha$-LEO, we have $y_{o'} \ge \alpha x_{i,o'}$ for every outcome $o\neq \perp$. Summing over all outcomes $o' \succeq_i \overline{o}_{i,\tau}$ gives
\[
\sum_{o' \in \mathcal{O}:\, o' \succeq_i \overline{o}_{i,\tau}} y_{o'}
\ge
\alpha
\sum_{o' \in \mathcal{O}:\, o' \succeq_i \overline{o}_{i,\tau}} x_{i,o'}
\ge
\alpha \cdot  \Pr(\randbud{} \ge \tau),
\]
which proves \eqref{eq:002}.

Notice that, to guarantee \eqref{eq:002} for all $\tau$, it suffices to verify the condition for $\tau$ equal to the price of each outcome, $\tau=p_{i,o}$ . Therefore, \eqref{eq:002} is equivalent to
\begin{equation}\label{eq:003}
\sum_{o' \in C : o' \succeq_i o} y_{o'}  
\ge
\alpha  \Pr(\randbud{} \ge \prices_{i,o}) \qquad \forall o \in \mathcal{O}\setminus \{\perp\}, \forall  i\in N .
\end{equation}


With this, we obtain the following lemma.
\begin{lemma}\label{lemma:feasible}
Any $\alpha$-LEO $(\x_{i=1}^n,\y,\prices_{i=1}^n)$ with income distribution $\randbud{}$  satisfies
\eqref{eq:000}, \eqref{eq:001}, and \eqref{eq:003}. 
\end{lemma}

Next, we establish the following result, which, together with Lemma~\ref{lemma:feasible}, implies Lemma~\ref{theo:maintech}.
\begin{lemma}\label{lemma:main}
Fix a random distribution $\randbud{}$, and let $(\y, \prices)$ be any pair satisfying \eqref{eq:000}, \eqref{eq:001}, and \eqref{eq:003}. Then, for any $\tau>0$, the joint distribution over outcome–set pairs $\widetilde{L}^\tau(\y, \prices)$ produced by dependent rounding satisfies conditions \ref{cnd1}, \ref{cnd2}, and \ref{cnd3} of Lemma~\ref{theo:maintech}.
\end{lemma}


\begin{proof}

Given $(\y,\prices)$, the argument proceeds in two steps.

First, to prove \ref{cnd2},  we need to show that for every realization \(o^* \sim \widetilde{L}^\tau(\y,\prices)\) and every outcome \(o\),
\begin{equation}\label{eq:0007}
\bigl|\{ i \in V^\tau(o^*) : o \succ_i o^* \}\bigr|
\le 
\left\lfloor 
\alpha n \cdot \frac{c(o)}{B} \cdot \frac{\mathbb{E}(\randbud{})}{\tau} 
\right\rfloor .
\end{equation}

Let \(S := \{ i \in V^\tau(o^*) : o \succ_i o^* \}\).  
For each \(i \in S\), since \(i \in V^\tau(o^*)\), we have \(o^* \succeq_i \overline{o}_{i,\tau}\), and hence
\(o \succ_i \overline{o}_{i,\tau}\). By definition of \(\overline{o}_{i,\tau}\), this implies \(p_{i,o} > \tau\), and therefore
\[
\tau |S|
\le 
\sum_{i \in S} p_{i,o}
\le 
\sum_{i \in N} p_{i,o}
\le 
\alpha n \cdot \frac{c(o)}{B} \cdot \mathbb{E}(\randbud{}),
\]
where the last inequality follows from \eqref{eq:001}. Since \(|S|\) is integral, \eqref{eq:0007} follows.

Second, to prove \ref{cnd1},  we show that for any voter \(i\), the probability that \(i\) is not \(\tau\)-covered by a random outcome drawn from \(\widetilde{L}^\tau(\y,\prices)\) is small. In particular,
\begin{equation}\label{eq:005}
\Pr_{o^* \sim \widetilde{L}(\y,\prices)} \bigl( i \notin V^{\tau}(o^*,\prices) \bigr) 
\le 
e^{- \alpha \cdot \Pr(\randbud{} \ge \tau)}. 
\end{equation}

To see this, recall that 
\[
V^{\tau}(o^*,\prices) := \{ i \in N : o^* \succeq_i \overline{o}_{i,\tau} \}.
\]
If \(\overline{o}_{i,\tau} = \perp\), then \(i \in V^\tau(o^*)\) for every \(o^* \in \mathcal{O}\) as $\perp$ is the least-preferred outcome, and therefore \eqref{eq:005} holds trivially.  Thus, assume \(\overline{o}_{i,\tau} \neq \perp\).  
Since \((\y,\prices)\) satisfies \eqref{eq:003}, which is equivalent to \eqref{eq:002}, define the set
\[
\mathcal{O}_i^\tau := \{ o \in \mathcal{O} : o \succeq_i \overline{o}_{i,\tau} \}.
\]
Then \eqref{eq:002} can be written as
\[
\sum_{o \in \mathcal{O}_i^\tau} y_o \ge \alpha \cdot \Pr(\randbud{} \ge \tau).
\]

Next, observe that the event \(i \notin V^{\tau}(o^*,\prices)\) implies that none of the outcomes in \(\mathcal{O}_i^\tau\) is merged into \(o^*\) during the dependent rounding procedure. This further implies that \(y_o \le 1\) for all \(o \in \mathcal{O}_i^\tau\); otherwise, the rounding procedure would select \(o\) and merge it with probability one.

Applying negative dependence (Proposition~\ref{proposition:dep}), the probability that none of the outcomes in \(\mathcal{O}_i^\tau\) is merged into the final outcome $o^*$ is at most
\[
\prod_{o \in \mathcal{O}_i^\tau} (1-y_o)
\le 
\prod_{o \in \mathcal{O}_i^\tau} e^{-y_o}
\le 
e^{-\sum_{o \in \mathcal{O}_i^\tau} y_o}
\le 
e^{-\alpha \cdot \Pr(\randbud{} \ge \tau)},
\]

Finally, to prove \ref{cnd3}, observe that dependent rounding guarantees that the budget violation does not exceed the cost of the most expensive outcome, which completes the proof.

\end{proof}

%% file: new_deterministic.tex
The following theorem is the main result of this section.  

\begin{theorem}\label{thm:deterministic}
Given an instance of  budgeted social choice  \((\mathcal{O}, c(\cdot), \mathcal{C}, \succ, B)\) under Assumption~\ref{assu:merging}, there exists a \(6.24\)-approximate core solution.
\end{theorem}

For ease of presentation, we prove a slightly weaker bound of $6.67$ below. The proof of the improved $6.24$ bound is deferred to Appendix~\ref{sec:improved_approximation}.

We will construct the approximate core solution by incorporating the partial core solution into an iterative procedure. While the procedure is similar in its structure to \cite{ApproStable}, the use of LEO and the analysis framework based on the notion of partial core are novel. 

First, the non-emptiness of the random core implies the existence of a \textit{deterministic} outcome–voter pair constituting a partial-core solution via linearity of expectation.

\begin{lemma}\label{lemma:random-to-deterministic}
Given a $(\lambda,\gamma)$-randomized core, there exists a $\gamma$-partial core solution $(o^*, V^*)$ with $|V^*| \ge \lambda n$.
\end{lemma}

\begin{proof}
Given a $(\lambda,\gamma)$-randomized core, let $\widetilde{V}$ denote the induced lottery over voter sets. Since each voter is represented with probability at least $\lambda$, by linearity of expectation,
\[
\mathbb{E}_{V \sim \widetilde{V}}[|V|]
= \sum_{i \in N} \Pr_{V \sim \widetilde{V}}(i \in V)
\ge \sum_{i \in N} \lambda
= \lambda n.
\]
Therefore, there exists a realization $V^*$ of $\widetilde{V}$ such that $|V^*| \ge \lambda n$. Together with the corresponding outcome $o^*$ from the randomized core, the pair $(o^*,V^*)$ forms a $\gamma$-partial core with desired properties.
\end{proof}

By the existence of such deterministic partial core solutions, we introduce the notion of a \emph{partial core oracle}, which will be used in the procedure for constructing approximate core solutions. Specifically, we say that an oracle is a $(\lambda,\gamma)$-partial core oracle if, for any instance of budgeted social choice, it returns a $\gamma$-partial core solution $(o^*, V^*)$ satisfying $|V^*| \ge \lambda N$.

With the partial core oracle, we obtain the following result. The details of the proof is given in Appendix~\ref{proof:deterministic}.

\begin{lemma}\label{lemma:iterativepartial}
Given an instance of budgeted social choice satisfying Assumption~\ref{assu:merging}, let $\lambda \in (0,1)$ and $\gamma > 0$ be given, and suppose there exists a $(\lambda,\gamma)$-partial core oracle. Then, there exists an algorithm that makes a polynomial number of oracle calls and outputs an $(\Omega(\lambda), \gamma)$-approximate core solution, where
\[
\Omega(\lambda,\gamma)
:= \min_{\substack{\omega>1, \\ \omega (1-\gamma) < 1}}
\max \Biggl\{
\frac{\omega}{\omega-1} \cdot \frac{\gamma}{1-\omega(1-\lambda)}, \;\;
\frac{\omega}{\omega-1}
\Biggr\}.
\]
\end{lemma}

\noindent
\textbf{Proof of a weaker version of Theorem~\ref{thm:deterministic}.}
We obtain a $6.67$-approximate core solution as a direct consequence of Lemma~\ref{lemma:iterativepartial} and Theorem~\ref{theo:random}.  

Setting $\alpha=2.88$, Theorem~\ref{theo:random} implies that a $(\lambda,\gamma)$-randomized core always exists with
\[
\lambda = 1-e^{-2.88}, \qquad \gamma = 3.88.
\]
Applying Lemma~\ref{lemma:random-to-deterministic} yields a $(\lambda,\gamma)$-partial core oracle. With these values of $\lambda$ and $\gamma$, setting $\omega = 4.22$ satisfies $\omega(1-\lambda) < 1$. Lemma~\ref{lemma:iterativepartial} then implies the existence of a $6.67$-approximate core solution.

\qed

We next describe how the \emph{partial core oracle} can be used to construct approximate core solutions satisfying Lemma~\ref{lemma:iterativepartial}. Given any partial core solution $(o^*, V^*)$ and an alternative outcome $o$, voters can be partitioned into three categories: Category~1 consists of voters in $V^*$ who prefer $o^*$ to $o$; Category~2 consists of voters in $V^*$ who prefer $o$ to $o^*$; and Category~3 consists of voters in $N \setminus V^*$. The set of voters who prefer $o$ to $o^*$ is therefore contained in the union of Category~2 and Category~3 voters, both of which admit useful upper bounds from the $(\lambda,\gamma)$-partial core.

A key challenge is that while the number of Category~2 voters can be bounded proportionally to the cost $c(o)$, the number of Category~3 voters is only bounded by a constant fraction of the electorate. To address this, we apply an iterative procedure that treats Category~3 voters as a new instance and invokes the partial core oracle again. Repeating this process eventually eliminates Category~3 voters, so that the voters preferring $o$ to $o^*$ must be a subset of Category~2 voters across iterations.

\begin{figure}[h]
\centering 
\includegraphics[scale = 0.35]{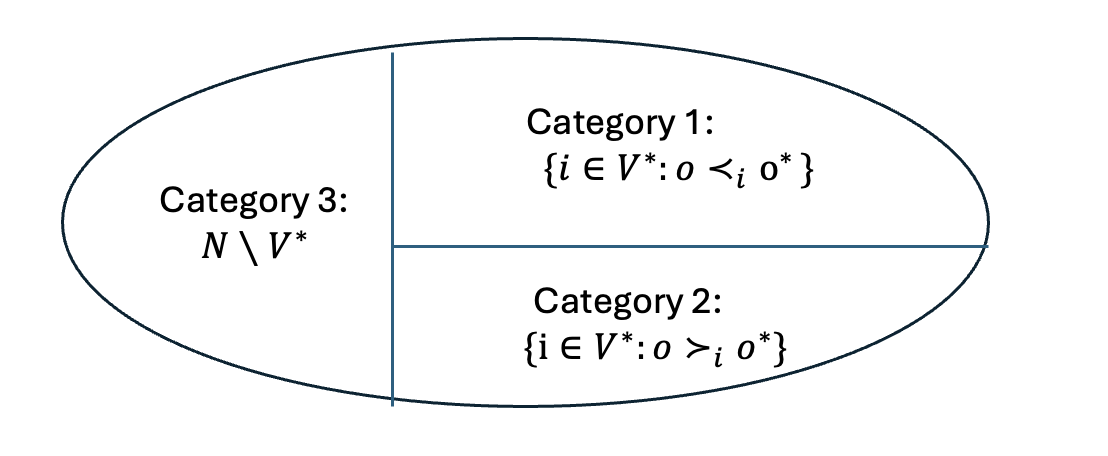}
\caption{Division of voters into three categories according to a partial core solution $(o^*,V^*)$ and alternative outcome $o \neq o^*$. }  
\end{figure}

\smallskip

The algorithm used to prove Lemma~\ref{lemma:iterativepartial} is given in Algorithm~\ref{alg:iterative}. Starting with a size $B_0$ and voters $V_0:=V$, at each round $t$, the algorithm invokes the $(\lambda,\gamma)$-oracle and obtains a $\lambda$-partial core solution $(o_t^*, V_t^*)$ with $|V_t^*|\geq \lambda\cdot |V_t|$ i.e. representing $\lambda$-proportion of the whole population. Then, we set $V_{t+1}=V_t\setminus V_t^*$ and $B_{t+1}=\frac{B_t}{\omega}$ and repeat the process on the  un-represented voters $V_{t+1}$ and outcomes with cost at most $B_{t+1}$. The procedure will terminate once every voter is represented, and the output is $\bigoplus_t o^*_t$. Then, given an alternative outcome $o$, the number of voters deviating to $o$ is upper bounded by the sum of the number of Category (2) voters accumulated across iterations, given that all Category (3) have been eliminated.

\begin{algorithm}[htb!]   
\KwData{a set $N$ of $n$ voters, an outcome space $\mathcal{O}$, a budget $B$, a $(\lambda,\gamma)$-partial core oracle with $\lambda\in (0,1)$ and $\gamma\geq 1$, parameters $\omega$ satisfying $\omega\cdot (1-\lambda)<1$} 
  \KwResult{an outcome $o^*$ with $c(o^*)\leq B$}
  $j\leftarrow 0, o^* \leftarrow \perp, V_0 \leftarrow N$, $B_0\leftarrow  \frac{\omega - 1}{\omega}\cdot B $\\
  \While{$|V_t| \geq 1$}{
        $\mathcal{O}^t\leftarrow \{o\in \mathcal{O}:c(o)\leq B_t\}$\\
        Given outcome space $\mathcal{O}^t$, voter set $V_t$ and budget $B_t$, find through the oracle a $\gamma$-partial core solution $(o^*_t, V^*_t)$ with $c(o^*_t)\leq B_t$ and $|V^*_t|\geq \lambda\cdot |V_t|$ \label{line4}
        \\
        $o^*\leftarrow o^*\oplus o^*_t$\\
        $B_{t+1}\leftarrow \frac{ B_t}{\omega} $ \label{line7}             \\
        $V_{t+1}\leftarrow V_t\setminus V^*_t$\\ 
        t++
}
  return $o^*$
  \caption{Iterated Rounding with Partial Core} \label{alg:iterative}
\end{algorithm}

To see why the algorithm outputs an approximate core solution, the key observation is that $B_t$ and $|V_t|$ both decrease \textit{geometrically}, with $B_t$ decreasing at a rate of $\frac{1}{\omega}$ and $|V_t|$ decreasing at a rate of $1-\lambda$. First, as $B_t$ and therefore $c(o^*_t)$ decreases \textit{geometrically}, the final merged outcome $o^*$ is guaranteed to be within the budget. Second, when ensuring that $|V_t|$ decreases faster than $B_t$, the number of Category (2) voters at each stage will also decrease \textit{geometrically} and therefore the total number of deviating voters can also be bounded.

%% file: newcomputation.tex
The previous sections established existence results. We now turn to computational considerations, which present two main challenges.

The first challenge concerns the size of the outcome space. When the number of outcomes is exponential, computation can become intractable.\footnote{For example, the algorithm in \cite{ApproStable}, which computes approximate stable lotteries, and thus approximate core solutions, is polynomial-time only when the number of outcomes is polynomial.} We face a similar difficulty here. To enable efficient computation under general preferences, some additional structure must be assumed.

A direct assumption that the number of outcomes is polynomial makes computing approximate deterministic core outcomes straightforward (although computing randomized core outcomes remains nontrivial). Alternatively, one can assume that, even when the outcome space is exponential, the set of relevant core constraints is polynomial. This assumption is natural in many applications: either the problem inherently limits comparisons to a polynomial number of alternative outcomes, or the underlying structure reduces the set of relevant comparisons to a polynomial size (see Section~\ref{sec:conclude} for such applications). We adopt this assumption in this section.

Even under this restriction, a second challenge remains. Our construction relies on LEO, which requires solving a fixed-point problem. We show that this difficulty can be avoided. The key observation is that the properties of LEO used to construct core solutions can be characterized by a linear program when LEO uses a particular income distribution—namely, the uniform distribution. This characterization allows the resulting algorithm to run in polynomial time.

To formalize this idea, we consider a restricted notion of the (approximate) core, in which each voter compares the selected outcome only against a subset of outcomes \(\mathcal{C} \subseteq \mathcal{O}\). We refer to \(\mathcal{C}\) as the \emph{comparison set}, and our goal is to design an algorithm whose running time is polynomial in \(|\mathcal{C}|\).

We accordingly redefine the core, approximate core, and randomized core with respect to a comparison set \(\mathcal{C} \subseteq \mathcal{O}\). For notational convenience, we assume that the null outcome \(\perp \in \mathcal{C}\).

\begin{definition}[Core notions with respect to a comparison set $\mathcal{C}$]\label{def:core_compact}
Let $(\mathcal{O}, c(\cdot), \{\succeq_i\}_{i=1}^n, B)$ be an instance of budgeted social choice and $\mathcal{C} \subseteq \mathcal{O}$ a comparison set.

\begin{itemize}
    \item 
    A feasible outcome  $o^*$ is a  \emph{$\gamma$-approximate core solution with respect to $\mathcal{C}$} if
    \[
        |\{i \in N : o^* \prec_i o \}| \le \gamma \cdot \frac{c(o)}{B} \cdot n \quad \forall o \in \mathcal{C}.
    \]

    \item A pair $(o^*, V)$ with $V \subseteq N$ is a \emph{$\gamma$-partial core solution with respect to $\mathcal{C}$} if
    \[
        |\{ i \in V : o^* \prec_i o \}| \le \gamma \cdot \frac{c(o)}{B} \cdot n \quad \forall o' \in \mathcal{C}.
    \]
    A voter $i$ is \emph{represented} by the partial-core solution if $i \in V$.

    \item A distribution over $\gamma$-partial core solutions is a \emph{$(\lambda,\gamma)$-randomized core with respect to $\mathcal{C}$} if every voter $i$ is represented with probability at least $\lambda$.
\end{itemize}
\end{definition}

For the remainder of this section, we fix the comparison set $\mathcal{C}$, and to avoid clutter, in this section, we will sometimes omit the phrase “with respect to $\mathcal{C}$.” The following results constitute the main computational contributions of this section.

\begin{theorem}\label{thm:computation0}  
Given an instance of budgeted social choice $(\mathcal{O}, c(\cdot), \succ, B)$ satisfying Assumption~\ref{assu:merging}, and a comparison set $\mathcal{C} \subseteq \mathcal{O}$, for any $\tau \in (0,1)$ there exists a polynomial-time algorithm (in $n$ and $|\mathcal{C}|$) that computes a $(\lambda,\gamma)$-randomized core solution with respect to $\mathcal{C}$, where
\[
\lambda = 1 - e^{-\alpha (1-\tau)}, \qquad 
\gamma = \frac{\alpha+1}{2\tau}.
\]
\end{theorem}

\begin{theorem}\label{thm:computation}
Given an instance of budgeted social choice $(\mathcal{O}, c(\cdot), \succ, B)$ satisfying Assumption~\ref{assu:merging}, and a comparison set $\mathcal{C} \subseteq \mathcal{O}$, there exists a polynomial-time algorithm (in $n$ and $|\mathcal{C}|$) that computes an $11.6$-approximate core solution with respect to $\mathcal{C}$.
\end{theorem}

\begin{proof}[Proof of Theorem~\ref{thm:computation0}]

Our proof follows the same high-level approach as before, based on the Lindahl Equilibrium with Ordinal preferences (LEO). The main difference is that we now use a uniform income distribution, $\randbud{} \sim \mathbf{U}[0,1]$, which corresponds to the setting where the outcome set is $\mathcal{C}$ rather than $\mathcal{O}$.

The key idea is to use Lemma~\ref{lemma:main}, which shows that a randomized core can be constructed from any $(\y, \prices)$ satisfying \eqref{eq:000}, \eqref{eq:001}, and \eqref{eq:003}. Moreover, the rounding algorithm that constructs the randomized core from $(\y, \prices)$ can be implemented in polynomial time (Proposition~\ref{proposition:dep}).

An important observation that leads to computational tractability is that, when $(\y, \prices)$ are treated as unknowns, \eqref{eq:000} and \eqref{eq:001} are linear constraints. Furthermore, if $\randbud{} \sim \mathbf{U}[0,1]$, \eqref{eq:003} is also linear, since $\Pr(\randbud{} \ge p_{i,o}) = 1 - p_{i,o}$. Hence, these inequalities can be formulated as the following linear program:

\begin{align}
\text{LP}(\alpha,\mathcal{C},B): \quad & p_{i,o} \in [0,1], \ y_o \in [0,B] && \forall i \in N, o \in \mathcal{C} \setminus \{\perp\} \nonumber\\
& \sum_{o \in \mathcal{C}} c(o) y_o = B && \label{lp-eq1}\\
& \sum_{i \in N} p_{i,o} \le \frac{\alpha}{2} \cdot \frac{c(o)}{B} \cdot n && \forall o \in \mathcal{C} \setminus \{\perp\} \label{lp-eq2}\\
& \sum_{o' \in \mathcal{C} : o' \succeq_i o} y_{o'} \ge \alpha \cdot (1 - p_{i,o}) && \forall i \in N, o \in \mathcal{C}\setminus \{\perp\} \label{lp-eq3}
\end{align}

First, this LP has a feasible solution, as any $\alpha$-LEO with $\randbud{} \sim \mathbf{U}[0,1]$ satisfies it (Lemma~\ref{lemma:feasible}). Second, the LP has polynomial number of  constraints, and can therefore be solved in $\mathrm{poly}(n,|\mathcal{C}|)$ time.  
Solving this LP yields a feasible $(\y, \prices)$ satisfying the required conditions.

Next, as in the proof of Theorem~\ref{theo:random}, define 
\[
\mathcal{C}^* := \{ o \in \mathcal{C} : c(o) \le \tfrac{1}{\alpha+1} B \}, 
\qquad 
B^* := \frac{\alpha}{\alpha+1} B.
\]
We apply Lemma~\ref{lemma:main} to a solution $(\y,\prices)$ of LP$(\alpha,\mathcal{C}^*,B^*)$. The corresponding lottery $\widetilde{L}^\tau(\y,\prices)$ then yields a $(\lambda,\gamma)$-randomized core solution with respect to $\mathcal{C}$ and original budget $B$, where
\[
\lambda = 1 - e^{-\alpha\cdot\Pr(\randbud{}\ge \tau)}
= 1 - e^{-\alpha(1-\tau)},
\qquad
\gamma = \frac{(\alpha+1)\cdot \mathbb{E}[\randbud{}]}{\tau}
= \frac{\alpha+1}{2\tau}.
\]

\end{proof}

\begin{proof}[Proof of Theorem~\ref{thm:computation}]

Given Theorem~\ref{thm:computation0} with $\alpha = 6.57$ and $\tau = 0.495$,  we can obtain a poly-time $(\lambda,\gamma)$-oracle with $\lambda = 0.964,\gamma = 7.64.$ The, we plug $\lambda$, $\gamma$ and $\omega = 5.11$ into Lemma~\ref{lemma:iterativepartial} and obtain the desired approximation factor. 

\end{proof}

%% file: conclusion.tex
This paper introduces a novel Lindahl equilibrium defined using only ordinal rankings. We show how this concept can be used to construct both randomized and deterministic core solutions for a broad class of monotonic preferences.

While our analysis begins from a normative perspective, treating equilibrium as a benchmark for desirable outcomes, it also provides practical tools: polynomial-time algorithms that implement these outcomes in certain settings. These tools have immediate applications, which we illustrate below, and we expect that future work will extend their relevance to participatory decision-making, resource allocation, and machine learning settings.

\paragraph{\textit{Voting with Ranking Ballots}}
An implication of Theorem~\ref{thm:computation} is that approximate core solutions can be computed efficiently for participatory budgeting problems with ranking ballots. Ranking ballots have been studied in \citep{ChengJiangMunagalaWang2019, ApproStable, proportional21} and are widely used in participatory budgeting instances.

In this setting, each voter provides a ranking over the available projects. Preferences over sets of projects are naturally extended from these rankings: a voter prefers set $S$ over $S'$ if her top-ranked project in $S$ is ranked higher than her top project in $S'$. Each project has an associated cost, and the cost of a set is simply the sum of its projects’ costs.

Theorem~\ref{thm:computation} implies that, for an instance with $m$ projects, $n$ voters, and a total budget $B$, there exists a polynomial-time algorithm that produces a $11.6$-approximate core solution. The key observation is that, because preferences depend only on a voter’s top project, it suffices to consider a comparison set consisting of individual projects. This reduces the size of the comparison set to a polynomial number, making the computation tractable.

\paragraph{\textit{Proportional Fair clustering}}
Fair clustering is a central problem in machine learning. \cite{pmlr-v97-chen19d} introduce \emph{proportional fairness} for centroid clustering: given $n$ data points and $k$ clusters, a clustering is proportionally fair if no group of at least $n/k$ points can all strictly reduce their distances by switching to another center. This setting can be captured in our framework, where data points have rankings over centers, the set of outcomes consists of $k$-element subsets of centers, and the consideration set contains singletons.

\cite{pmlr-v97-chen19d} show that exact proportionally fair clusterings may not exist and develop algorithms that compute \emph{approximate} proportional fairness. In particular, they provide an efficient algorithm guaranteeing $(1+\sqrt{2})$-approximate fairness, meaning that no sufficiently large coalition can simultaneously reduce its distances by more than a factor of $1+\sqrt{2}$.

Our results generalize this approach. We do not require a metric space or relax comparisons between points and centers; instead, we bound the number of points that can deviate to a better center. Theorem~\ref{thm:computation0} provides a randomized solution with a clear interpretation. For example, setting $\alpha = 2$ and $\tau = 1/2$ yields a polynomial-time algorithm that produces a distribution over $k$ clusters and corresponding subsets of data such that, in each realization, no subgroup of $\frac{\alpha+1}{2\tau} \, n/k = 3n/k$ points can improve by switching to a different center. Moreover, across the distribution, each data point appears in these subsets with probability at least $1 - e^{-\alpha(1-\tau)} \approx 0.63$.\footnote{In fact, for this problem, we can provide a stronger result. Since the cost of each center is 1, dependent rounding incurs no budget violation, and we can achieve a $(1 - e^{-\alpha(1-\tau)}, \frac{\alpha}{2\tau})$-randomized core in polynomial time.} Meanwhile, Theorem~\ref{thm:computation} implies that our polynomial-time algorithm achieves a $\gamma = 11.6$ approximation, meaning that at most a $11.6\frac{n}{k}$  data points can simultaneously switch to a  better center.

\paragraph{\textit{Fair Multi-Label Classification}}
Fair classification assigns data points to discrete labels under constraints that prevent systematic disadvantage across individuals or groups, and has been widely applied in lending, hiring, and medical prediction \citep{Barocas2019FairnessML}.

We consider a generalization where each of the $n$ data points can be associated with a small subset of labels rather than a single label. Given $m$ possible labels, the goal is to select $k$ representative labels and assign each data point to at most $\Delta \ll k$ labels. Allowing multiple labels provides a compact representation that captures diverse features more accurately than a single label \citep{ZhangZhou2014}. For instance, an article on “climate policy, renewable energy, and economic incentives” can be labeled with \{\emph{climate, energy, economy}\}, capturing all key aspects; similarly, a movie like \emph{The Lord of the Rings} can be labeled with \{\emph{fantasy, adventure, epic}\}, reflecting its main characteristics.

This setting fits naturally into our framework: each data point ranks all combinations of up to $\Delta$ labels, forming the comparison set, while the set of outcomes is all label subsets. The strength of this approach lies in representing preferences over meaningful label combinations. Although the outcome space is large, the comparison set has size at most $O(m^\Delta)$, allowing polynomial-time computation for any fixed $\Delta$. A multi-label classification solution consists of a selected set of $k$ labels and an assignment of each data point to its most preferred combination of at most $\Delta$ labels. The selection and assignment is \emph{proportionally $\gamma$-fair} if no combination of $\Delta$ labels—including at least one label outside the selected set—is strictly preferred by more than a $\gamma \Delta n / k$ fraction of data points.

Theorems~\ref{thm:computation0} and \ref{thm:computation} provide polynomial-time approximation algorithms for constructing randomized or deterministic solutions with proportional fairness guarantees.

%% file: newappendix1.tex
\section{Existences of LEO}\label{sec:proof-Leo}

\subsection{Proof of Lemma~\ref{lemma:conti}}\label{app:conti}
\begin{proof}
 Let $F_{\randbud{}}$ be the CDF for the budget given to an agent $i \in N$. We have that
\begin{align*}
    x'_{i,o} &= 
\Pr[\mathcal{D}_{i}(\prices_{i}, \randbud{}) =o] \\
&=  
\begin{cases}
    \underset{o' \in \mathcal{O}: o'\succ_i o}{\min}\bigg( F_{\randbud{}}(p_{i,o'}) - F_{\randbud{}}(p_{i,o})\bigg)^+ & \text{if } o \in \mathcal{O} \text{ is not the top-ranked in } \succ_i, \\
    1 - F_{\randbud{}}(p_{i,o}) & \text{if } o \in \mathcal{O} \text{ is the top-ranked in } \succ_i.
\end{cases}
\end{align*}
Intuitively, if $o$ is the top-ranked in $\succ_i$, then $i$ demands $o$ as long as $p_{i,o}$ is at most the budget given to $i$. The probability of having such a budget is $1 - F_{\randbud{}}(p_{i,o})$. If $o$ is the second-ranked and $o'$ is the top-ranked in $\succ_i$, then $i$ demands $o$ if the budget is below $p_{i,o'}$ and at least $p_{i,o}$. This probability is captured by $\left(F_{\randbud{}}(p_{i,o'}) - F_{\randbud{}}(p_{i,o})\right)^+$. Note that if $p_{i,o'} < p_{i,o}$, then $i$ never demands $o$ since it is more expensive and less preferred. Applying analogous reasoning results in the closed form above for the probability that $i$ demands $o$ under a random budget $\randbud{}$. More specifically, consider $o' \succ_i o$, the probability that $o$ is demanded by $i$ under $\prices_i$ is zero when there exists $p_{i,o'} \le p_{i,o}$. Otherwise, it is the minimum difference between the CDF at $p_{i,o'}$ and $p_{i,o}$. Since $F_{\randbud{}}$ and the closed form above are continuous, Lemma~\ref{lemma:conti} follows.
\end{proof}

\subsection{Proof of Theorem \ref{thm:existence}: Existence of LEO} \label{app:LEO}

We include the renowned Kakutani's fixed point theorem and Maximum theorem here for the sake of completeness. 
\begin{theorem}[Kakutani's fixed point theorem, \cite{Kakutani41}]\label{thm:kakutani}
Let $S$ be a non-empty, compact, and convex subset of some Euclidean space $\mathbb{R}^n.$ Let $\psi:S\rightarrow 2^S \setminus \emptyset$ be a point-to-set function on $S$ such that 1) $\psi$ is upper-hemicontinuous and 2) $\psi(s)$ is non-empty and convex for all $s \in S$. Then there exists a \emph{fixed point} $s \in S$ such that $s \in \psi(s)$. 
\end{theorem}

\begin{theorem}[Maximum Theorem, \cite{maximum}]\label{thm:maximum}
Given non-empty $\mathbf{X}\subset \mathbb{R}^L$ and $\mathbf{\Theta}\subset \mathbb{R}^M$, let $f:\mathbf{X}\times \mathbf{\Theta}\rightarrow \mathbb{R}$ be a continuous function on the product $\mathbf{X}\times\mathbf{\Theta}$, and $\mathbf{C}:\mathbf{\Theta}\rightrightarrows \mathbf{X}$ be a compact valued correspondence such that $\mathbf{C}(\theta)\neq \emptyset$ for all $\theta \in\mathbf{\Theta}$. Define $f^*(\theta) = \sup\{f(x,\theta):x\in \mathbf{C}(\theta)\}$ and the correspondence $\mathbf{C}^*: \mathbf{\Theta} \rightrightarrows \mathbf{X}$ by $\mathbf{C}^*(\theta) = \{x \in \mathbf{C}(\theta): f(x,\theta) = f^*(\theta)\}$. If $\mathbf{C}$ is continuous at $\theta$, then $f^*$ is continuous and $\mathbf{C}^*$ is upper hemi-continuous, non-empty, and compact valued. As a consequence, the $\sup$ can be replaced by $\max$. 
\end{theorem}

Equipped with these theorems, we are ready to show Theorem \ref{thm:existence}.

\begin{proof}

Define $\Gamma:= \{\z \in \mathbb{R}_+^{|\mathcal{O}\setminus \{\perp\}|} \mid \underset{o\in \mathcal{O}\setminus \{\perp\}}{\sum} c_oz_o = B \} $. We construct the following correspondence:

$$\mathcal{L}: \prod_{i=1}^n[0,1]^{|\mathcal{O}|}\times \Gamma \times \prod_{i=1}^n[0,1]^{|\mathcal{O}|} \rightrightarrows \prod_{i=1}^n[0,1]^{|\mathcal{O}|} \times \Gamma  \times\prod_{i=1}^n[0,1]^{|\mathcal{O}|}$$
with
$$\mathcal{L}\left((\{\x_{i}\}_{i\in N},\y,\{\prices_{i}\}_{i\in N})\right) = (\{\x'_{i}\}_{i\in N},Y',\{\prices'_{i}\}_{i\in N})$$ where

\begin{subequations}
\begin{align}
x'_{i,o} &= \Pr[\mathcal{D}_{i}(\prices_{i}, \randbud{}) =o] & \forall i \in N, o \in \mathcal{O}, \label{eq:fp-x} \\
Y' &= \arg \max_{\z \in \Gamma} \;\; \left(\sum_{i\in N} \prices_i\right)^T \z, \label{eq:fp-y} \\ 
p'_{i,o} &= 
\begin{cases}
\max\left\{\min\left\{1,p_{i,o}+(x'_{i,o}-y_o)\right\} ,0 \right\} \text{  if } o \neq \emptyset \\
0 \text{  if } o = \perp,
\end{cases} & \forall i \in N, o \in \mathcal{O}. \label{eq:fp-p}
\end{align}
\end{subequations}
Our goal is to show that
\begin{enumerate}
    \item \label{pf-leo-1} By Theorem \ref{thm:maximum}, the correspondence $\prod_{i=1}^n[0,1]^{|\mathcal{O}|} \rightrightarrows \Gamma$ (from $\prices$ to $Y'$) is upper hemi-continuous, non-empty, and compact valued. Furthermore, $Y'$ is convex, bounded, and non-empty.
    \item \label{pf-leo-2} By Theorem \ref{thm:kakutani}, there exists a \emph{fixed point} $(\{\x_{i}\}_{i\in N},\y,\{\prices_{i}\}_{i\in N})$ such that 
    $$(\{\x_{i}\}_{i\in N},\y,\{\prices_{i}\}_{i\in N }) \in \mathcal{L}(\{\x_{i}\}_{i\in N},\y,\{\prices_{i}\}_{i\in N}).$$
    \item \label{pf-leo-3} The fixed point $(\{\x_{i}\}_{i\in N},\y,\{\prices_{i}\}_{i\in N})$ satisfies all the conditions of a LEO in Definition \ref{def:LEO}. 
\end{enumerate}

To show \ref{pf-leo-1}, we employ Theorem \ref{thm:maximum} with $\mathbf{X} = \Gamma$ (the space of the producer's decision) and $\mathbf{\Theta} = \prod_{i=1}^n[0,1]^{|\mathcal{O}|}$ (the space of prices $\prices$), both of which are non-empty. Let $f(\z,\prices) = \left(\sum_{i \in N} \prices_i\right)^T \z$ be the producer revenue given the producer decision $\z \in \Gamma$ and the prices $\prices \in \mathbf{\Theta}$. Given that for any prices $\prices \in \mathbf{\Theta}$, any producer decision in $\Gamma$ is feasible, we set $\mathbf{C}(\prices) = \Gamma$ for any $\prices \in \mathbf{\Theta}$. $\mathbf{C}$ is a compact valued correspondence since $\Gamma$ is compact. $f^*$ maps $\prices$ to the maximum producer revenue, thus, given prices $\prices$, we have $\mathbf{C}^*(\prices) = \arg \max_{\z \in \Gamma} \;\; \left(\sum_{i\in N} \prices_i\right)^T \z$. That is, $\mathbf{C}^*(\prices)$ is the set of the optimal producer decision. By Theorem \ref{thm:maximum}, because $\mathbf{C}$ is continuous at any $\prices \in \mathbf{\Theta}$, $f^*$ is continuous and $\mathbf{C}^*$ is upper hemi-continuous, non-empty, and compact valued. Furthermore, $\mathbf{C}^*(\prices)$ is convex, non-empty, and bounded for any $\prices \in \mathbf{\Theta}$ since it is the set of the optimal solutions of a linear program with a bounded and non-empty feasible region.

To show \ref{pf-leo-2}, we consider $S = \prod_{i=1}^n[0,1]^{|\mathcal{O}|}\times \Gamma \times \prod_{i=1}^n[0,1]^{|\mathcal{O}|}$ and let $\psi = \mathcal{L}$. The range of $\psi$ is restricted to having $\x'$ and $\prices'$ as a point and $Y' \subseteq \Gamma$ as a set. Clearly, $S$ is non-empty, compact, and convex. For all $s \in S$, $\psi(s)$ is non-empty and convex since $Y'$ is convex, non-empty, and bounded. 

The remaining is to show that $\mathcal{L}$ is upper hemi-continuous.  To show this, we prove three properties: (a) the mapping from $s \in S$ to $\x'$ is continuous, which follows from Lemma~\ref{lemma:conti}; (b) the mapping from $s \in S$ to $Y'$ is upper hemi-continuous, which directly follows from \ref{pf-leo-1}; and (c) the mapping from $s \in S$ to $\prices'$ is continuous, which holds once (a) is established, since $p'_{i,o}$ defined in \eqref{eq:fp-p} is a continuous function of $\x'$, $\y$, and $\prices$.

Thus, by Theorem \ref{thm:kakutani}, there exists $s \in S$ such that $s \in \mathcal{L}(s)$.

To show \ref{pf-leo-3}, we prove that any fixed point $s=(\{\x_{i}\}_{i\in N},\y,\{\prices_{i}\}_{i\in N})$ such that $s \in \mathcal{L}(s)$ satisfies conditions \ref{leo-1}, \ref{leo-2}, and \ref{leo-3} in Definition \ref{def:LEO}. Conditions \ref{leo-1} and \ref{leo-3} follow directly by \eqref{eq:fp-x} and \eqref{eq:fp-y}, respectively. For condition \ref{leo-2}, we first show that $x_{i,o} \le y_o$ for all $o \in \mathcal{O}$ and $i \in N$. Suppose for the sake of contradiction that $x_{i,o} > y_o$ for some $o \in \mathcal{O}$ and $i \in N$. Since $x_{i,o} - y_o > 0$ and $s$ is a fixed point, from \eqref{eq:fp-p}, we must have $p_{i,o} = 1$. Consequently, $i$ cannot afford $a$ with a strictly positive probability unless $\Pr[\randbud{}=1]$ is strictly positive. However, because $\randbud{}$ is a random variable with support on the unit interval $[0,1]$, a strictly positive $\Pr[\randbud{}=1]$ would imply that $F_{\randbud{}}$ is discontinuous, which is a contradiction. Therefore, $x_{i,o} \le y_o$. Now suppose $x_{i,o} < y_o$, then from \eqref{eq:fp-p}, the fixed point $s$ forces $p_{i,o} = 0$. Hence, $s=(\{\x_{i}\}_{i\in N},\y,\{\prices_{i}\}_{i\in N})$ satisfies all the conditions of Definition \ref{def:LEO}, so $s$ is a LEO. 
\end{proof}

%% file: proofs.tex
\section{Missing Proofs}

\subsection{Bounding revenue-cost ratio}\label{app:random}
\begin{proposition} \label{prop:upper_bound_price}
Let an  $\alpha$-LEO $(\x,\y,\prices)$ with an income distribution $\mathcal{I}$ supported on $[0,1]$ be given. Then, it holds that for every $o\in\mathcal{O}$, the revenue-to-cost ratio can be upper bounded as: 
\[
\frac{\sum_{i=1}^n p_{i,o}}{c(o)}\;\le\; \alpha n\cdot \frac{\mathbb{E}(\mathcal{I})}{B}.
\]
\end{proposition}

\begin{proof} 
For each voter $i$, revenue extracted from $i$ equals $\alpha$ times her total spending, because $y_o\ne \alpha x_{i,o}$ only when $p_{i,o}=0$. Hence
\[
\sum_{o\in\mathcal{O}} p_{i,o}y_o
= \sum_{o : y_o=\alpha x_{i,o}} p_{i,o}y_o
= \alpha \sum_{o\in\mathcal{O}} p_{i,o} x_{i,o}.
\]
Since $i$'s spending is at most her income,
\[
\sum_{i=1}^n \sum_{o\in\mathcal{O}} p_{i,o}y_o
\;\le\; \alpha \sum_{i=1}^n \mathbb{E}(\mathcal{I})
= \alpha\, n\, \mathbb{E}(\mathcal{I}). \tag{$\star$}
\]

Suppose, toward contradiction, that some $o'$ satisfies
$
\sum_{i=1}^n p_{i,o'} \;>\; 
\alpha\, n\, \mathbb{E}(\mathcal{I})\cdot \frac{c(o')}{B},
$
so its revenue-to-cost ratio exceeds $\alpha n\cdot \mathbb{E}(\mathcal{I})/B$.  
By the revenue-maximization property of $y$, every outcome in $\mathrm{supp}(y)$ must then have strictly larger revenue-to-cost ratio than $\frac{\alpha n\mathbb{E}(\mathcal{I})}{B}$ .  
Thus
\[
\sum_{i=1}^n\sum_{o\in\mathcal{O}} p_{i,o}y_o
=\sum_{o\in\mathcal{O}} \frac{\sum_i p_{i,o}}{c(o)}\cdot c(o)y_o
> \frac{\alpha n\mathbb{E}(\mathcal{I})}{B}\sum_{o\in\mathcal{O}} c(o)y_o
= \alpha\, n\, \mathbb{E}(\mathcal{I}),
\]
contradicting $(\star)$.  
\end{proof}

\subsection{Proof of Lemma~\ref{lemma:iterativepartial}}\label{proof:deterministic}

\begin{proof}
Suppose the input of Algorithm~\ref{alg:iterative} satisfies $\omega(1-\lambda)<1$. We show that the output $o^*$ is an $f(\lambda,\gamma,\omega)$-core solution with $c(o^*) \le B$, where
\[
f(\lambda,\gamma,\omega)
:= \max \left\{
\frac{\omega}{\omega-1} \cdot \frac{\gamma}{1-\omega(1-\lambda)},
\frac{\omega}{\omega-1}
\right\}.
\]

We first verify feasibility. The output $o^*$ satisfies
\[
c(o^*) \le \sum_t c(o^t) \le \sum_t B_t
= B_0 \sum_{t=0}^{\infty} \frac{1}{\omega^t} =\frac{\omega}{\omega - 1}\cdot B_0
= B.
\]

Fix an arbitrary outcome $o \in \mathcal{O}$. If $c(o) > B_0$, then the number of voters who prefer $o$ to $o^*$ is at most $n$, which satisfies
\[
n \le \left(\frac{B}{B_0}\right)\frac{c(o)}{B} \cdot n.
\]
Since $\frac{B}{B_0} = \frac{\omega}{\omega-1}$, the desired bound follows.

We now consider outcomes $o$ with $c(o) \le B_0$. For each round $t$ of Algorithm~\ref{alg:iterative}, let
\[
K_{o,t} := \{ i \in V_t^* : o \succ_i o_t^* \}.
\]
We claim that
\[
|K_{o,t}| \le [\omega(1-\lambda)]^t \cdot \gamma \cdot \left(\frac{c(o)}{B_0} \cdot n\right).
\]

Indeed, for any $o$, the partial core guarantee implies
\[
|K_{o,t}| \le \gamma \cdot \frac{c(o)}{B_t} \cdot |V_t|.
\]
This holds trivially when $c(o) \ge B_t$ (since $\gamma \ge 1$), and follows directly from the definition of the partial core when $c(o) \le B_t$.

From the algorithm, $|V_{t+1}| = (1-\lambda)|V_t|$, implying
\[
|V_t| = (1-\lambda)^t |V_0| = (1-\lambda)^t n,
\]
and $B_{t+1} = \frac{B_t}{\omega}$, implying
\[
B_t = \frac{B_0}{\omega^t}.
\]
Substituting these expressions yields the claimed bound on $|K_{o,t}|$.

Finally, summing across rounds,
\[
\big| \{ i \in N : o \succ_i o^* \} \big|
\le \sum_{t=0}^{\infty} |K_{o,t}|
\le \sum_{t=0}^{\infty} [\omega(1-\lambda)]^t \cdot \gamma \cdot \left(\frac{c(o)}{B_0} \cdot n \right),
\]
which evaluates to
\[
\gamma \cdot \frac{\omega}{\omega-1} \cdot \frac{1}{1-\omega(1-\lambda)}
\cdot \left(\frac{c(o)}{B} \cdot n \right).
\]

This establishes the desired core bound and completes the proof.
\end{proof}

\section{6.24-Approximation for General Preferences and Costs}\label{sec:improved_approximation}

To obtain the improved approximation factor, we introduce a slightly modified version of partial core, which also takes into account the cost of the outcome to compare with. This will offer us enhanced flexibility in controlling the convergence rate of the iterative algorithm. 
\begin{definition}\label{def:approximate3}
Given an instance of budgeted social choice \((\mathcal{O}, c(\cdot), \{\succeq_i\}_{i=1}^n, B)\), a pair \((o^*, V)\) consisting of an feasible outcome $o^*$ and a subset of voters $V \subseteq N$ is in the \underline{$(\gamma,\rho)$-partial core} if, for every outcome $o \in \mathcal{O}$ with $c(o)\le \frac{B}{\rho}$,
\[
|\{i \in V : o^* \prec_i o\}| \leq \gamma \cdot \frac{c(o)}{B} \cdot n.
\]
Given a voter $i$, we say that $i$ \underline{is represented} by the partial $(\gamma,\rho)$-core solution $(o^*, V)$ if $i \in V$.
\end{definition}

The notions of partial core oracle and random core are also adjusted accordingly. 

\begin{definition}
We say that oracle is a $(\lambda,\gamma,\rho)$-partial core oracle if, for any instance of budgeted social choice, it returns a $(\gamma,\rho)$-partial core solution $(o^*, V^*)$ satisfying $|V^*| \ge \lambda N.$
\end{definition}

We then prove the existence of a partial core oracle with desirable parameters.

\begin{lemma}\label{lem:partial_oracle_improved}
For any $\rho>1$ and $\alpha>0$, there exists a $(1-e^{-\alpha},\frac{\alpha\cdot \rho}{\rho-1},\rho)$-partial core oracle. 
\end{lemma}
\begin{proof}
Let $B^* = B\cdot \frac{\rho-1}{\rho}$ and $O^* = \{o\in \mathcal{O}: c(o)\leq \frac{B}{\rho}\}$. Similar to the proof of Theorem \ref{theo:random}, we can apply Lemma \ref{theo:maintech} to a $\alpha$-LEO with respect to $O^*$ and $B^*$ to obtain a lottery s.t. (1) each voter is represented with a probability of $1-e^{-\alpha}$; (2) the number of people in $V^*$ preferring an alternative outcome $o$ to $o^*$ is at most $\alpha\cdot \frac{c(o)}{\frac{\rho}{\rho-1}\cdot B}\cdot n$; (3) each realized outcome is of cost at most $\frac{(\rho-1)\cdot B}{\rho}+\frac{B}{\rho} = B$. Condition (2) and (3) together mean that every realized outcome-voters pair correspond to a $(\frac{\alpha\cdot\rho}{\rho-1},\rho)$-partial core for outcome space $\mathcal{O}$ with a budget of $B$. Then, the proof follows in the same way from linearity of expectation as in the proof of Lemma \ref{lemma:random-to-deterministic}. 
\end{proof}

Next, we slightly modify the iterative rounding algorithm based on the new partial core oracle. All we need is to replace $\gamma$-partial core (oracle) with $(\gamma,\rho)$-partial core (oracle).

\begin{algorithm}[htb!]   
\KwData{a set $N$ of $n$ voters, an outcome space $\mathcal{O}$, a budget $B$, a $(\lambda,\gamma,\rho)$-partial core oracle with $\lambda\in (0,1)$ and $\gamma,\rho\geq 1$, parameters $\omega$ satisfying $\omega\cdot (1-\lambda)<1$} 
  \KwResult{an outcome $o^*$ with $c(o^*)\leq B$}
  $j\leftarrow 0, o^* \leftarrow \perp, V_0 \leftarrow N$, $B_0\leftarrow  \frac{\omega - 1}{\omega}\cdot B $\\
  \While{$|V_t| \geq 1$}{
        $\mathcal{O}^t\leftarrow \{o\in \mathcal{O}:c(o)\leq B_t\}$\\
        Given outcome space $\mathcal{O}^t$, voter set $V_t$ and budget $B_t$, find through the oracle a $(\gamma,\rho)$-partial core solution $(o^*_t, V^*_t)$ with $c(o^*_t)\leq B_t$ and $|V^*_t|\geq \lambda\cdot |V_t|$  
        \\
        $o^*\leftarrow o^*\oplus o^*_t$\\
        $B_{t+1}\leftarrow \frac{ B_t}{\omega} $ 
        $V_{t+1}\leftarrow V_t\setminus V^*_t$\\ 
        t++
}

  return $o^*$
  \caption{Iterated Rounding with Partial Core: Version 2} \label{alg:iterative_improved}
\end{algorithm}

While the iterative algorithm is mostly the same, the analysis is more involved because, at each stage, we no longer have a single formula bounding the number of deviating voters for \textit{all} outcomes. 

\begin{lemma}\label{lemma:bound_each_step_1}
In round $t$ of Algorithm~\ref{alg:iterative_improved}, for any outcome $o$ with $c(o)\leq \frac{B_t}{\rho}$, it holds:
$$|K_{o,t}| \leq [\omega\cdot (1-\lambda)]^t \cdot \gamma \cdot (\frac{c(o)}{B_0}\cdot n).$$
\end{lemma}
\begin{proof}
The argument follows the proof of Lemma~\ref{lemma:iterativepartial}. For any $o$, the partial core guarantee gives
\[
|K_{o,t}| \le \gamma \cdot \frac{c(o)}{B_t} \cdot |V_t|.
\]
This holds trivially when $c(o) \ge B_t$ (since $\gamma \ge 1$), and follows from the partial core definition when $c(o) \le B_t$.

From the algorithm, $|V_t|=(1-\lambda)^t n$ and $B_t=B_0/\omega^t$. Substituting these expressions yields the desired bound on $|K_{o,t}|$.
\end{proof}

We then establish the correctness of the algorithm. A key tweak of the proof is to show that while we cannot guarantee $|K_{o,t}|$ through Lemma \ref{lemma:bound_each_step_1} if $c(o)>\frac{B_t}{\rho}$, we can show that the number of remaining voters have become very small, to the point that it is no greater than the sum of the remaining terms in the geometric series, so that we can still obtain a desirable bound.

\begin{lemma}\label{lemma:correctness_improved}
If the input of Algorithm \ref{alg:iterative_improved} satisfy $\omega\cdot(1-\lambda)<1$ and $\frac{\gamma} {\rho}\geq  1-\omega\cdot (1-\lambda)$, then the output of the algorithm is a $f(\lambda,\gamma,\omega)$-core solution with 
$$f(\lambda,\gamma,\omega):= \gamma\cdot  \frac{\omega}{\omega -1}\cdot \frac{1}{1-\omega\cdot (1-\lambda)}.$$
\end{lemma}
\begin{proof}
We first see that the output $o^*$ is within the budget: 
$$c(o^*)\leq \sum_t c(o^t)\leq \sum_t B_t = B_0\cdot \sum_{t}\frac{1}{\omega^t} = B.$$

Now, pick an arbitrary outcome $o\in \mathcal{O}$. If $c(o)> \frac{B_0}{\rho}$, then it necessarily follows that the number of voters deviating from $o^*$ to $o$ is at most $(\frac{B}{B_0/\rho})\cdot \frac{c(o)}{B}\cdot n $, where $\frac{B}{B_0/\rho} = \frac{\omega\cdot \rho}{\omega-1}\leq f(\lambda,\gamma,\omega)$ under the assumption. 

Suppose otherwise, we let $\bar{t}$ be the biggest $t$ for which $c(o)\leq \frac{B_{\bar{t}}}{\rho}$. Then, we can bound
\setcounter{equation}{0}
\begin{equation}\label{eq2}
\big|i\in V:  o\succ_i o^*\big|\leq \big(\sum_{t=0}^{\bar{t}}|K_{o,t}|\big)+ |V_{\bar{t}+1}| .
\end{equation}

To prove (\ref{eq2}), we claim that 
\begin{equation}\label{eq3}
|V_{\bar{t}+1}|\leq \gamma \cdot \frac{c(o)}{B}\cdot n\cdot \sum_{t=\bar{t}+1}^\infty \bigg(\omega\cdot (1-\lambda)\bigg)^t.
\end{equation}

Plugging (\ref{eq3}) and Lemma \ref{lemma:bound_each_step_1} into (\ref{eq2}) yields 
\begin{align*}
|i\in V: o\succ_i o^* | & \leq \big(\sum_{t=0}^{\bar{t}}|K_{o,t}|\big)+ |V_{\bar{t}+1}| \\
&\leq \gamma \cdot (\frac{c(o)}{B_0}\cdot n)\cdot \sum_{t=0}^\infty\bigg(\omega\cdot (1-\lambda)\bigg)^t\\
& \leq \gamma \cdot (\frac{c(o)}{B_0}\cdot n)\cdot \frac{1}{1-\omega\cdot (1-\lambda)}\\
& = \gamma\cdot \frac{\omega}{\omega -1}\cdot \frac{1}{1-\omega\cdot (1-\lambda)}\cdot (\frac{c(o)}{B_0}\cdot n).
\end{align*}

Then, it is left for us to verify (\ref{eq3}) using the assumption $\omega\cdot(1-\lambda)<1$ and $\frac{\gamma}{\rho}\leq  1-\omega\cdot (1-\lambda)$:
\begin{align*}
|V_{\bar{t}+1}| 
&\leq \frac{1}{(1-\lambda)^{\bar{t}+1}} \cdot n 
= \frac{B_0}{\omega^{\bar{t}+1}} \cdot \frac{n}{B_0} \cdot [\omega\cdot (1-\lambda)]^{\bar{t}+1} \\
& =B_{\bar{t}+1} \cdot \frac{n}{B_0} \cdot [\omega\cdot (1-\lambda)]^{\bar{t}+1}\\
&\leq  \frac{\gamma/\rho}{1-\omega\cdot (1-\lambda)}\cdot B_{\bar{t}+1}\cdot \frac{n}{B_0} \cdot [\omega\cdot (1-\lambda)]^{\bar{t}+1}\\ 
&= \gamma \cdot \frac{B_{\bar{t}+1}/\rho}{B_0}\cdot n\cdot \frac{[\omega\cdot (1-\lambda)]^{\bar{t}+1}}{1-\omega\cdot (1-\lambda)}\\
&\leq \gamma \cdot \frac{c(o)}{B}\cdot n\cdot \sum_{t=\bar{t}+1}^\infty \bigg(\omega\cdot (1-\lambda)\bigg)^t,
\end{align*}
and this completes the proof.
\end{proof}

\begin{proof}[Proof of Theorem \ref{thm:deterministic}:] In this case, we set $\alpha = 2.88$ and $\rho = 4.88$, from which we obtain a $(1-e^{-2.88}, \frac{2,88*4.88}{3.88},4.88)$-partial core by \ref{lem:partial_oracle_improved}. We input the oracle and $\omega = 4.6$ into the algorithm and obtain a 6.24 approximate core solution by Lemma \ref{lemma:correctness_improved}.
\end{proof}